\documentclass[aps,prresearch,twocolumn,amsmath,amssymb,longbibliography]{revtex4-2}

\usepackage{amsthm}
\usepackage{bm}
\usepackage{graphicx}
\graphicspath{{figures/}}
\usepackage{hyperref}
\usepackage{booktabs}
\usepackage{xcolor}
\usepackage{mathtools}

\newtheorem{proposition}{Proposition}
\newtheorem{corollary}{Corollary}
\newtheorem{remark}{Remark}

\DeclareMathOperator{\Var}{Var}
\DeclareMathOperator{\Cov}{Cov}

\begin{document}

\title{Hierarchical Disorder in Moir\'e Exciton Photoluminescence Probed by Spectral-Descriptor Correlations}

\author{Katsunori Wakabayashi}
\email{WAKABAYASHI.Katsunori@nims.go.jp}
\affiliation{Research Center for Materials Nanoarchitectonics (MANA), 
National Institute for Materials Science (NIMS),
Namiki 1-1, Tsukuba 305-0044, Japan}

\date{\today}

\begin{abstract}
Hyperspectral photoluminescence (PL) maps of moir\'e transition-metal
dichalcogenide heterobilayers encode rich information about the underlying
disorder, but extracting it is hampered by the ambiguity of assigning individual
spectral peaks. We show that the spatial correlations among a set of simple,
peak-decomposition-free spectral descriptors---such as the centroid energy, the
dominant-peak energy, and a sharp-line fraction---provide a peak-decomposition-free
route to infer features of that disorder. Different descriptors act as filters that select
different components of a multi-scale disorder landscape---a smooth,
micron-correlated background and a dense set of localized traps---and therefore
acquire different spatial correlation lengths. The central prediction is a
correlation-length hierarchy, $\xi(E_{\rm cent}) \ge \xi(E_{\rm dom})$, which we
derive by splitting the dominant-peak energy into a smooth background part and a
short-range trap-switching fluctuation. The same picture explains the measured
inter-descriptor correlations, including the near-perfect anticorrelation
$\rho_{\rm S}(\Delta E_{\rm cd}, R_{\rm HL}) \approx -0.978$, which we show to be
a robust geometric trend for spectra dominated by a common emission-envelope asymmetry.
Benchmarked against phenomenological simulations, Hamiltonian diagonalization, and
the measured MoSe$_2$/WSe$_2$ descriptor correlations, the framework turns
descriptor maps into a quantitative, peak-decomposition-free probe of slow disorder
and local traps in moir\'e excitons and, more broadly, in disordered semiconductor
emitters.
\end{abstract}

\maketitle

\section{Introduction}
\label{sec:intro}

Moir\'e heterostructures formed by vertically stacking two-dimensional materials
with a small twist angle or lattice mismatch have become a central platform
for engineering quantum states of electrons and excitons
\cite{andrei2021marvels,mak2022semiconductor,regan2022emerging,geim2013vdw,novoselov2016heterostructures,
cao2018unconventional,cao2018correlated,kennes2021moire}.
In transition metal dichalcogenide (TMD) heterobilayers---building on the strong
photoluminescence of monolayer TMDs \cite{mak2010atomically,splendiani2010emerging,eda2011photoluminescence}
and their rich optical and valley physics \cite{wang2018colloquium,mak2016photonics,schaibley2016valleytronics}---
the long-wavelength moir\'e potential reorganizes interlayer excitons into
localized minibands, enabling moir\'e-trapped exciton physics
including moir\'e Bose-Hubbard models, exciton crystals,
and quantum emitter arrays
\cite{yu2017moire,tran2019evidence,jin2019observation,xu2020correlated,
tang2020simulation,shimazaki2020strongly,
rivera2015observation,rivera2016valley,nagler2017interlayer,okada2018direct,ciarrocchi2022excitonic}.

Among these systems, MoSe$_2$/WSe$_2$ heterobilayers are particularly
rich platforms for studying moir\'e-modulated optical responses.
At low temperatures, the interlayer exciton photoluminescence (PL) is
characteristically complex:
a broad emission envelope near $1.25$--$1.40~$eV coexists with dense,
narrow spectral peaks whose microscopic assignment remains debated,
with proposed contributions from moir\'e-trapped interlayer excitons,
phonon-assisted momentum-indirect emission, donor-acceptor pair excitons,
and disorder-localized states \cite{seyler2019signatures,alexeev2019resonantly,
jin2019observation,brem2020hybridized,choi2021twist,
he2015single,koperski2015single,srivastava2015optically,chakraborty2015voltage,tonndorf2015single}.

\begin{figure*}[t]
\centering
\includegraphics[width=\textwidth]{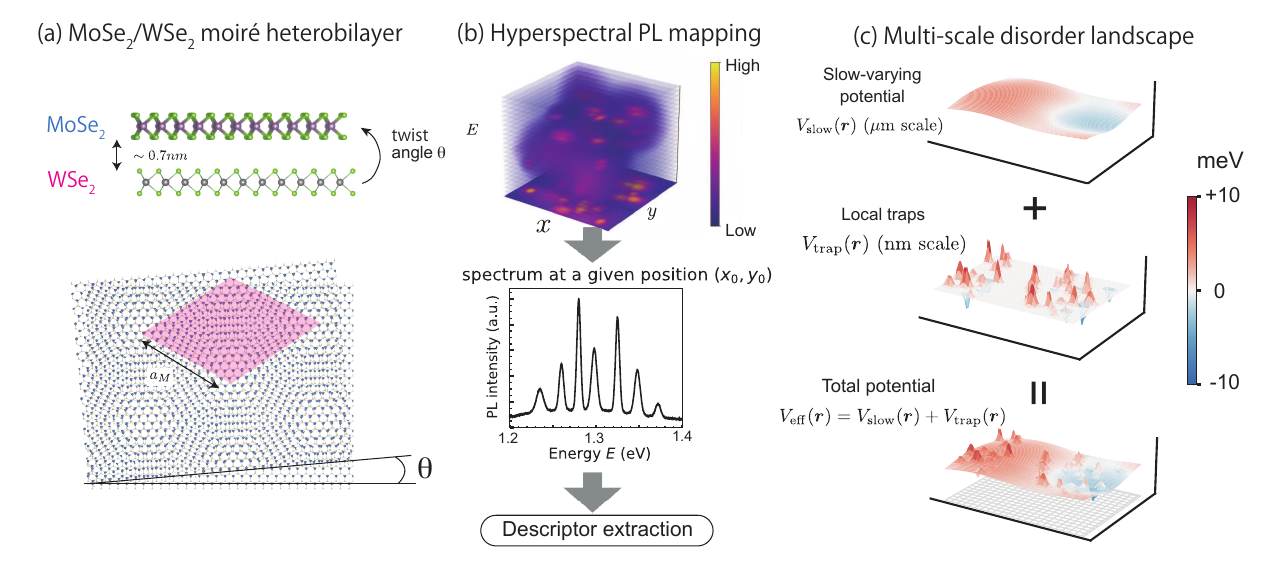}
\caption{
\textbf{Overview of the system, measurement, and theoretical framework.}
(a) MoSe$_2$/WSe$_2$ heterobilayer with interlayer distance
$\sim\!0.7$~nm and twist angle $\theta$,
producing a long-wavelength moir\'e superlattice of period $a_M$.
(b) Hyperspectral photoluminescence mapping yields
a spatially resolved PL spectrum at each map position $(x_0,y_0)$,
from which peak-decomposition-free spectral descriptors
($E_{\mathrm{cent}}$, $E_{\mathrm{dom}}$, $\Delta E_{cd}$, $W_{80}$,
$R_{\mathrm{HL}}$, $F_S$, $R_1$, $S_{\mathrm{spec}}$, $I_{\mathrm{tot}}$)
are extracted without microscopic line assignment.
(c) Theoretical multi-scale disorder landscape: a slowly varying
background potential $V_{\mathrm{slow}}(\mathbf{r})$ at the micron scale
combined with localized traps $V_{\mathrm{trap}}(\mathbf{r})$ at sub-optical length scales
yields the micron-scale effective potential
$V_{\mathrm{eff}}(\mathbf{r}) \simeq V_{\mathrm{slow}}(\mathbf{r}) + V_{\mathrm{trap}}(\mathbf{r})$
that organizes the spectral descriptors hierarchically; the moir\'e term in the
full potential [Eq.~\eqref{eq:Veff}] is coarse-grained over the optical spot and
enters only through the renormalized mass and local density of states.
}
\label{fig:overview}
\end{figure*}

A key observation motivating the present work is that
the spectral complexity is \emph{spatially organized}.
Figure~\ref{fig:overview} sketches the system:
a twisted MoSe$_2$/WSe$_2$ moir\'e heterobilayer [Fig.~\ref{fig:overview}(a)],
its hyperspectral PL map [Fig.~\ref{fig:overview}(b)], and the multi-scale
disorder landscape underlying our model [Fig.~\ref{fig:overview}(c)].

Recent hyperspectral PL mapping experiments by Ahmad et al.\ \cite{ahmad2025hierarchical}
on a MoSe$_2$/WSe$_2$ heterobilayer reveal that
nine peak-decomposition-free spectral descriptors
(centroid energy $E_{\mathrm{cent}}$, dominant-peak energy $E_{\mathrm{dom}}$,
centroid-dominant offset $\Delta E_{cd}$, quantile width $W_{80}$,
high/low ratio $R_{\mathrm{HL}}$, sharp fraction $F_S$,
roughness $R_1$, spectral entropy $S_{\mathrm{spec}}$,
and integrated intensity $I_{\mathrm{tot}}$)
show characteristic micron-scale spatial correlations.
Feature-wise spatial correlation analysis yields 1/e decay lengths of
$1.27$--$2.05~\mu$m for all descriptors,
exceeding the $0.85~\mu$m optical spot size.
Gaussian mixture clustering on the principal component analysis (PCA)-projected descriptor space identifies
three dominant spectral families that form contiguous real-space domains.
At the same time, individual pixels retain a dense multi-peak structure
that varies across the map,
implying an unresolved local spectral manifold below optical resolution.

Here, the \emph{local spectral manifold} denotes the optically unresolved
ensemble of local emitting states---including moir\'e-localized and
trap-related exciton states---that collectively form the PL spectrum
within a single optical spot; we use the term in this physical sense,
not in the differential-geometric one.
More broadly, hyperspectral PL imaging is emerging as a spatially resolved probe
of strain and disorder in TMD systems~\cite{alfrey2026revealing}.

This combination of phenomena---resolved micron-scale spectral landscape
coexisting with unresolved local-spectral-manifold complexity---suggests that a
single disorder scale is insufficient.
The central experimental conclusion is that
the emission is organized \emph{hierarchically}:
a slowly varying micron-scale background shapes the envelope,
while local spectral manifold degrees of freedom produce dense fine structure
within that background.

Despite extensive experimental and numerical studies of moir\'e exciton physics,
a minimal theoretical framework for hierarchical spectral inhomogeneity
as an organizing principle has not been established.
Previous theoretical work has focused primarily on:
(i) single-particle moir\'e miniband structures \cite{wu2018hubbard,wu2017theory,naik2022intralayer,ruiz2019interlayer,naik2018ultraflatbands},
(ii) correlated exciton and electron phases \cite{shimazaki2020strongly,tang2020simulation,regan2020mott,bai2020excitons,ghiotto2021quantum},
(iii) isolated defect-trapped excitons as quantum emitters \cite{parto2021defect,yu2021moire,branny2017deterministic},
or (iv) Gaussian disorder models for energy broadening without spatial structure \cite{brem2020hybridized}.
These approaches do not directly address the \emph{separation of scales}
observed in hyperspectral mapping.

In this paper, we develop a minimal theory for the spatial organization of
spectral descriptors in moir\'e exciton PL---minimal in the sense of containing
the fewest ingredients needed to reproduce the observed descriptor hierarchy,
not minimal in a microscopic sense---built around a single organizing principle:
\emph{different descriptors act as filters for different components of the disorder landscape}.
A descriptor that integrates spectral weight (centroid energy) averages over
local trap fluctuations and tracks the smooth, micron-scale potential $V_\mathrm{slow}$,
while a descriptor that locates spectral peaks (dominant-peak energy) is additionally
randomized by trap-level switching within the optical spot.
This differential filtering is why the descriptor-specific correlation lengths differ,
why the inter-descriptor Spearman matrix has the structure it does,
and why three spectral families emerge from a continuous disorder landscape.
It thereby connects the leading effective disorder parameters $(\xi_s, W_s, W_t)$
to the observable descriptor covariance structure, without any microscopic peak assignment.

The principal contributions of this work are the following.
(i) We formalize the descriptor-based disorder-filter principle introduced above into a quantitative mapping from the descriptor covariance matrix to the disorder parameters $(\xi_s, W_s, W_t)$ \cite{tarantola2005inverse}.
(ii) We derive a \textbf{correlation length hierarchy} $\xi(E_{\mathrm{cent}}) \geq \xi(E_{\mathrm{dom}})$ from the decomposition $E_{\mathrm{dom}} = \bar{V}_s + \delta E_{\mathrm{dom}}$. The trap-switching fluctuation $\delta E_{\mathrm{dom}}$ adds short-range variance, which systematically shortens the correlation length of the dominant-peak energy relative to that of the centroid.
(iii) We show that the nontrivial content of this hierarchy lies not in the mere
fact that the centroid is smoother than the peak position (expected of any
average), but in its magnitude: the ratio $\xi(E_{\mathrm{cent}})/\xi(E_{\mathrm{dom}})$
quantitatively encodes the relative strength of the trap and slow disorder and the
trap density. We turn this into an explicit, peak-decomposition-free protocol
(Table~\ref{tab:param_extract}) that constrains the leading effective disorder
parameters from four measured descriptor observables---the centroid correlation
length, the centroid variance, the $E_{\mathrm{cent}}$--$E_{\mathrm{dom}}$
correlation, and the mean sharp fraction (the last fixing the trap density only
after calibration of the sharp-fraction scale)---applicable, in principle, to
hyperspectral PL maps containing both smooth envelopes and unresolved local fine structure.
(iv) We derive analytically the sign and approximate magnitude of the principal inter-descriptor Spearman correlations. This establishes the $\Delta E_{cd}$--$R_{\mathrm{HL}}$ anti-correlation as a robust spectral-shape relation, valid for a broad class of spectra with a dominant unimodal envelope, and shows the $F_S$--$R_1$ co-variation to follow from both descriptors measuring the same trap-peak density.
(v) We provide an independent model-level check via Hamiltonian diagonalization---without relying on the synthetic two-fluid spectrum construction---recovering the correlation-length hierarchy from the eigenvalue spectrum alone.
(vi) We construct a disorder parameter map that identifies four qualitatively distinct regimes and places the experimental system in the hierarchically disordered regime where both slow and trap disorder exceed the homogeneous linewidth.

Throughout, the theoretical framework is a \emph{spectral-statistical organization theory}:
the Hamiltonian defines a continuum spectral manifold whose coarse-grained statistics
are resolved by the optical spot, and the descriptors are the observables through which
the disorder hierarchy is encoded and constrained.
This is distinct from a quantum-transport or localization theory---the relevant
physics operates at the $\sim\!\mu$m scale of the optical envelope,
not at the nanometer scale of individual eigenstates.

The remainder of this paper is organized as follows.
Section~\ref{sec:model} introduces the effective Hamiltonian and disorder model.
Section~\ref{sec:spectral_theory} develops the Green's function formalism
and local PL spectrum.
Section~\ref{sec:descriptor_theory} derives the statistical properties
of spectral descriptors as disorder filters.
Section~\ref{sec:hierarchy} derives the correlation length hierarchy relation.
Section~\ref{sec:correlations} derives inter-descriptor Spearman correlations.
Section~\ref{sec:phase_diagram} maps the four disorder regimes.
Section~\ref{sec:localization} presents the local-density-of-states (LDOS) two-fluid structure.
Section~\ref{sec:families} interprets spectral families as correlated disorder
domains, and Section~\ref{sec:glass} characterizes the multi-scale spectral
organization that emerges in the hierarchical regime.
Section~\ref{sec:numerics} describes Hamiltonian diagonalization and
phenomenological simulations.
Section~\ref{sec:discussion} discusses limitations and parameter extraction.
Section~\ref{sec:conclusion} concludes.

\section{Model: Effective Exciton Hamiltonian}
\label{sec:model}

This section constructs the effective Hamiltonian on which the rest of the paper
rests. We model the interlayer exciton as a single center-of-mass particle moving
in an effective potential, and assemble that potential from three physically
distinct ingredients---the moir\'e superlattice, a slowly varying long-wavelength
disorder, and sparse local traps---each introduced in turn below. We then identify
the separation of length scales among these ingredients, which is what makes the
descriptor-filter analysis of the following sections tractable.

\subsection{Center-of-Mass Hamiltonian}

We describe the interlayer exciton by its center-of-mass (COM) coordinate
$\mathbf{r}$ in the moir\'e plane.
The internal exciton structure is integrated out, yielding an effective
single-particle problem for the COM motion:

\begin{equation}
H = -\frac{\hbar^2}{2M}\nabla^2 + V_{\mathrm{eff}}(\mathbf{r}),
\label{eq:H}
\end{equation}
where $M$ is the interlayer exciton effective mass.
For MoSe$_2$/WSe$_2$, the electron and hole are spatially separated
into different layers \cite{rivera2015observation,nagler2017interlayer},
giving $M \approx m_e + m_h \approx 1.0$--$1.5~m_0$
\cite{berkelbachtheory,wang2018colloquium,chernikov2014exciton}
where $m_0$ is the free electron mass.

The effective potential is decomposed into three components
acting on different length scales:

\begin{equation}
V_{\mathrm{eff}}(\mathbf{r})
= V_{\mathrm{moir\acute{e}}}(\mathbf{r})
+ V_{\mathrm{slow}}(\mathbf{r})
+ V_{\mathrm{trap}}(\mathbf{r}).
\label{eq:Veff}
\end{equation}
These three components describe physics at three distinct length scales,
as summarized in Table~\ref{tab:scales}.

\begin{table*}[t]
\centering
\begin{ruledtabular}
\begin{tabular}{lccc}
Component & Length Scale & Physical Origin & Role \\
\colrule
$V_{\mathrm{moir\acute{e}}}$ & $a_M \sim 10$--$20~\mathrm{nm}$ & Moir\'e superlattice & Miniband formation \\
$V_{\mathrm{slow}}$ & $\xi_s \sim 1$--$2~\mu\mathrm{m}$ & Strain, twist-angle gradient & Spectral domain formation \\
$V_{\mathrm{trap}}$ & $\sigma_t \lesssim \sigma_{\mathrm{opt}}$ (few--100 nm effective) & Localized traps (defects/strain/charge) & Sharp line generation \\
\end{tabular}
\end{ruledtabular}
\caption{Three disorder components and their characteristic length scales.
Here $\sigma_t$ is an effective trap width spanning the experimentally debated
sharp-emitter origins (atomic defects, strain, or electrostatic disorder),
not a microscopic defect size; the theory requires only $\sigma_t \lesssim \sigma_{\mathrm{opt}} \ll \xi_s$.}
\label{tab:scales}
\end{table*}

\subsection{Moir\'e Potential}

The moir\'e superlattice potential is modeled by the
lowest-order hexagonally symmetric Fourier components:

\begin{equation}
V_{\mathrm{moir\acute{e}}}(\mathbf{r})
= 2V_m \sum_{j=1}^{3} \cos(\mathbf{G}_j \cdot \mathbf{r} + \phi_j),
\label{eq:Vmoire}
\end{equation}
where $\mathbf{G}_j$ are the three inequivalent moir\'e reciprocal lattice vectors
satisfying $|\mathbf{G}_j| = 4\pi/(a_M\sqrt{3})$,
$V_m$ controls the modulation amplitude,
and $\phi_j$ encodes the local stacking registry.
For small twist angle $\theta$ and lattice constant $a_0$,
the moir\'e period is $a_M \approx a_0/\theta$.

For a twist angle $\theta \sim 2$--$4^\circ$, $a_M \sim 10$--$20~$nm,
and the moir\'e potential depth is estimated at $V_m \sim 10$--$50~$meV
from first-principles calculations \cite{wu2018hubbard,wu2017theory,naik2022intralayer,ruiz2019interlayer,yu2017sciAdv}.

This term confines excitons to moir\'e unit cells, creating
localized miniband-like states on the nanometer scale.
For the purposes of spatial PL statistics at the micron scale,
the moir\'e potential can be treated as a fast-varying background
that renormalizes the exciton effective mass and density of states (DOS),
but does not directly set the micron-scale spatial correlations.

\subsection{Long-Wavelength Correlated Disorder}
\label{sec:Vslow}

The dominant component governing micron-scale spectral organization
is a slowly varying, spatially correlated random potential:

\begin{equation}
\langle V_{\mathrm{slow}}(\mathbf{r}) \rangle = 0,
\qquad
\langle V_{\mathrm{slow}}(\mathbf{r})\, V_{\mathrm{slow}}(\mathbf{r}') \rangle
= W_s^2\, \mathcal{C}_s(|\mathbf{r} - \mathbf{r}'|),
\label{eq:Vslow_corr}
\end{equation}
with Gaussian correlation kernel:

\begin{equation}
\mathcal{C}_s(r) = \exp\!\left(-\frac{r^2}{2\xi_s^2}\right).
\label{eq:Cslow}
\end{equation}
The characteristic parameters are:
$W_s \sim 5$--$20~$meV (disorder amplitude)
and $\xi_s \sim 1$--$2~\mu$m (correlation length),
much larger than the moir\'e period $a_M$.

The physical origins of $V_{\mathrm{slow}}$ are diverse.
Spatial variations in the local twist angle $\delta\theta(\mathbf{r})$, directly imaged in twisted TMD bilayers \cite{weston2020atomic,mcgilly2020visualization}, shift the interlayer exciton energy as $\partial E/\partial\theta \cdot \delta\theta$ (the twist-angle dependence of interlayer-exciton properties in TMD heterobilayers is well established \cite{choi2021twist}).
Inhomogeneous strain fields $\varepsilon_{ij}(\mathbf{r})$ couple through the deformation potential $\sim D_{ij}\varepsilon_{ij}$ \cite{conley2013bandgap}.
Electrostatic potential fluctuations arise from trapped charges or dielectric inhomogeneity \cite{raja2019dielectric}.
Finally, moir\'e reconstruction domain walls introduce step-like modulations \cite{weston2020atomic,mcgilly2020visualization,shabani2021deep}.
All of these generically produce smooth fields with correlation lengths
set by the macroscopic sample quality, typically $1$--$5~\mu$m
in high-quality van der Waals assemblies.

We work in the Gaussian approximation \eqref{eq:Vslow_corr}
as the minimal model, realized numerically by the spectral (FFT) method of
Appendix~\ref{app:GRF}; extensions to non-Gaussian slow disorder
are discussed in Section~\ref{sec:discussion}.

\subsection{Local Trap Disorder}
\label{sec:Vtrap}

Local trap potentials are introduced as a sum of attractive Gaussians
centered at random positions $\{\mathbf{R}_i\}$:

\begin{equation}
V_{\mathrm{trap}}(\mathbf{r})
= \sum_{i} u_i\, \exp\!\left(-\frac{|\mathbf{r} - \mathbf{R}_i|^2}{2\sigma_t^2}\right),
\label{eq:Vtrap}
\end{equation}
where $\mathbf{R}_i$ are drawn from a Poisson point process with density $n_t$; $u_i \leq 0$ are random trap depths drawn from a distribution $P(u)$ with mean $\bar{u} < 0$ and characteristic depth scale $W_t$ (the simulations realize $P(u)$ as a half-normal, $u_i = -|\mathcal{N}(0,W_t)|$, so that $W_t$ is the scale of the underlying normal, with second moment $\langle u^2\rangle = W_t^2$ and variance $W_t^2(1-2/\pi)$); and $\sigma_t \ll \xi_s$ is the trap radius.

Each trap acts as a local potential minimum that can bind an exciton
and produce a sharp spectral line.
The trap correlation function is:

\begin{equation}
\langle V_{\mathrm{trap}}(\mathbf{r})\, V_{\mathrm{trap}}(\mathbf{r}') \rangle
= n_t \langle u^2 \rangle \cdot \pi\sigma_t^2\,
\exp\!\left(-\frac{|\mathbf{r} - \mathbf{r}'|^2}{4\sigma_t^2}\right),
\label{eq:Vtrap_corr}
\end{equation}
with correlation length $\sim\sqrt{2}\sigma_t \ll \xi_s$.
This spatial white-noise-like character of the trap disorder
is what distinguishes it from $V_{\mathrm{slow}}$.
Microscopic sources of such localized potentials include native point defects
\cite{komsa2015native,parto2021defect},
strained nanobubbles and strain-localized emitters \cite{branny2017deterministic,koperski2015single},
and adsorbate-induced potential wells.

\subsection{Parameter Hierarchy and Scale Separation}

The effective model assumes a clear separation of length scales,
\begin{equation}
\sigma_t,\ a_M \ll \xi_s \ll L,
\label{eq:hierarchy}
\end{equation}
where $L$ is the sample (or map) size.
In the experiment~\cite{ahmad2025hierarchical}, $L \sim 8~\mu$m and
$\xi_s \sim 2~\mu$m. The two microscopic scales of the model lie far below
$\xi_s$: the moir\'e period $a_M \sim 10$--$20~$nm, and the trap radius $\sigma_t$,
which ranges from the atomic/defect scale (a few nm) up to $\sim$100~nm in our
robustness tests.

We deliberately leave the relative ordering of $\sigma_t$ and $a_M$ unspecified.
Instead, $\sigma_t$ is treated as an \emph{effective} potential width rather than a
microscopic trap dimension, because the origin of the sharp emitters---moir\'e-trapped
excitons, atomic defects, or strain and electrostatic disorder---remains
experimentally debated \cite{seyler2019signatures,brem2020hybridized}. These
possible origins have characteristic sizes ranging from a few nm up to $\sim$100~nm,
so no single microscopic value of $\sigma_t$ is preferred.

The theory requires only $\sigma_t \lesssim \sigma_{\mathrm{opt}} \ll \xi_s$
(with $\sigma_{\mathrm{opt}}$ the optical-spot radius), so that each trap acts as a
sub-spot localized emitter, with many such emitters contributing to a single
measured spectrum. With
this separation, $\sigma_t$ enters the descriptor statistics only through the
product $n_t W_t^2 \sigma_t^2$. This combination is not arbitrary: it is the
variance of the trap potential $V_{\mathrm{trap}}$. Summing the squared
contribution of each well over the Poisson-distributed trap positions gives the
shot-noise variance (Campbell's theorem~\cite{campbell1909discontinuous,campbell1910light})
$\Var V_{\mathrm{trap}} \propto n_t W_t^2 \sigma_t^2$, with the
three factors being the areal density $n_t$, the depth variance $W_t^2$, and the
well area $\sim\sigma_t^2$.
We identify it as the effective trap disorder $\Gamma_t$ in
Sec.~\ref{sec:phase_diagram}, and the descriptor variances that depend on it are
derived in Appendix~\ref{app:variances}.

Because $\sigma_t$ appears only inside this product, its value is not pinned down
on its own: any change in $\sigma_t$ can be offset by a compensating change in the
trap density ($n_t \propto \sigma_t^{-2}$), or in the depth variance $W_t^2$, that
keeps the product---and hence the leading trap-loading contribution to the
descriptor variances---approximately unchanged. In other
words, $\sigma_t$, $n_t$, and $W_t^2$ are \emph{degenerate}: they enter only in
this one combination, so the predictions are insensitive to the precise value of
$\sigma_t$ within the allowed window. The only place
$\sigma_t$ sets an independent length scale is the hierarchy floor
$\xi_t \equiv \sqrt{2}\,\sigma_t$ of Eq.~\eqref{eq:hierarchy_theorem}; even
for $\sigma_t \sim 100~$nm this remains well below $\xi_s$ ($\xi_t/\xi_s \sim 0.07$)
and therefore does not bind the hierarchy. The correlation-length hierarchy and the inter-descriptor
correlations are thus controlled primarily by $(\xi_s, W_s, W_t, n_t)$, with only
weak dependence on $\sigma_t$ within the sub-optical range;
a diagonalization sweep over $\sigma_t = 50$--$125~$nm confirms this explicitly
(Appendix~\ref{app:sigma_robust}).

This separation \eqref{eq:hierarchy} justifies treating the disorder components
as statistically independent and analytically tractable at each scale.
Figure~\ref{fig:disorder_landscape} shows a representative realization of the resulting
hierarchical disorder landscape: the slow potential $V_{\mathrm{slow}}$
[Fig.~\ref{fig:disorder_landscape}(a)], the local trap potential $V_{\mathrm{trap}}$
[Fig.~\ref{fig:disorder_landscape}(b)], and the effective potential
$V_{\mathrm{eff}} = V_{\mathrm{slow}} + V_{\mathrm{trap}}$
[Fig.~\ref{fig:disorder_landscape}(c)].

\begin{figure*}[t]
\centering
\includegraphics[width=\textwidth]{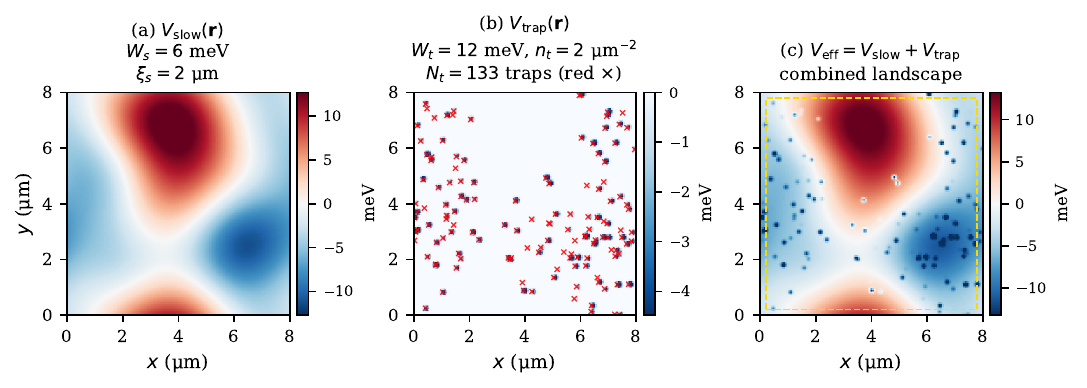}
\caption{
\textbf{Hierarchical disorder landscape.}
(a) Slow correlated disorder $V_{\mathrm{slow}}(\mathbf{r})$ with amplitude
$W_s = 6~$meV and correlation length $\xi_s = 2~\mu$m, generated by
FFT-based Gaussian filtering.
The smooth red/blue domains of size $\sim\xi_s$ represent the slow-disorder
component that gives rise to micron-scale spectral families in the model.
(b) Trap disorder $V_{\mathrm{trap}}(\mathbf{r})$ consisting of
$N_t = 133$ Gaussian potential wells ($W_t = 12~$meV, $\sigma_t = 50~$nm,
$n_t = 2~\mu$m$^{-2}$, consistent with $n_t L^2 \approx 128$ over the
$8\times 8~\mu$m$^2$ domain).
The color scale is saturated at the 2nd percentile of $V_{\mathrm{trap}}$ for
visibility; because the field is near zero between the sparse wells, this
displayed lower bound is much shallower than the deepest individual well
centers, whose bare trap-depth distribution is parameterized by $W_t = 12~$meV.
Red crosses mark the trap centers, which are drawn from
$P(\mathbf{r}) \propto \exp(-V_{\mathrm{slow}}(\mathbf{r})/E_{\mathrm{bias}})$
with $E_{\mathrm{bias}} = 7~$meV, biasing the trap positions, or optically active trap sites, toward slow-potential minima.
(c) Effective potential $V_{\mathrm{eff}} = V_{\mathrm{slow}} + V_{\mathrm{trap}}$,
showing the combined landscape experienced by excitons.
The gold dashed rectangle indicates the $20\times 20$ measurement grid area
(spanning $7.6\times 7.6~\mu$m$^2$, i.e.\ a $0.4~\mu$m pitch matching the
experimental $400$~nm step, with a $0.2~\mu$m margin from each edge).
}
\label{fig:disorder_landscape}
\end{figure*}

\section{Green's Function Formalism and Local PL Spectrum}
\label{sec:spectral_theory}

The spectral descriptors introduced above are all computed from the \emph{local}
PL spectrum $I(E,\mathbf{R})$ --- the emission collected from an optical spot
centered at position $\mathbf{R}$. The goal of this section is to express that
observable directly in terms of the disorder Hamiltonian $H$ of
Sec.~\ref{sec:model}, so that the statistics of the descriptors can later be
traced back to the underlying disorder parameters.

The retarded Green's function is the natural tool for expressing
$I(E,\mathbf{R})$ in terms of $H$. It resolves $H$ into
its spectral content, and its imaginary part yields the local density of states
(LDOS): the spectral weight that the eigenstates of $H$ place at energy $E$ and
position $\mathbf{r}$. The measured local spectrum is then this LDOS filtered by
the finite optical spot and weighted by the radiative coupling of each eigenstate
to light. In this way the formalism converts the eigenstates $\{\psi_n, E_n\}$,
which encode the multi-scale disorder, into the observable $I(E,\mathbf{R})$
\emph{without} fitting or identifying individual spectral peaks --- the
peak-decomposition-free approach on which the descriptor analysis is based.
We build this chain in three steps: the Green's function, the LDOS it generates,
and the radiatively weighted, optically filtered local spectrum.

\subsection{Retarded Green's Function}

The retarded Green's function is the resolvent of the exciton Hamiltonian: its
poles sit at the eigenenergies $E_n$ and its residues carry the eigenfunctions
$\psi_n$, so it packages the entire spectral content of $H$ in a single object
\cite{haug2009quantum},

\begin{equation}
G(\mathbf{r},\mathbf{r}';E)
= \langle \mathbf{r} | (E + i\eta - H)^{-1} | \mathbf{r}' \rangle
= \sum_n \frac{\psi_n(\mathbf{r})\psi_n^*(\mathbf{r}')}{E - E_n + i\eta},
\label{eq:Green}
\end{equation}
where $\psi_n(\mathbf{r})$ and $E_n$ are the eigenstates and eigenenergies of $H$,
and $\eta \to 0^+$ is an infinitesimal broadening.

\subsection{Local Density of States}

The local density of states (LDOS) is:

\begin{equation}
\rho(E,\mathbf{r}) = -\frac{1}{\pi}\,\mathrm{Im}\, G(\mathbf{r},\mathbf{r};E)
= \sum_n |\psi_n(\mathbf{r})|^2\, \delta(E - E_n).
\label{eq:LDOS}
\end{equation}
In practice, the delta function is replaced by a Lorentzian or Gaussian
with linewidth $\eta \sim 1$--$3~$meV (phenomenological exciton lifetime broadening):

\begin{equation}
\rho_\eta(E,\mathbf{r})
= \sum_n |\psi_n(\mathbf{r})|^2\, g_\eta(E - E_n),
\quad
g_\eta(E) = \frac{\eta/\pi}{E^2 + \eta^2}.
\label{eq:LDOS_broadened}
\end{equation}

\subsection{Local PL Spectrum}

The connection to experiment is set by the radiative coupling of the
center-of-mass (COM) exciton eigenstates to light. Because the internal
$1s$ relative wave function is essentially rigid, the light couples to each COM
eigenstate $\psi_n(\mathbf{r})$ through its \emph{coherent} dipole (oscillator-strength)
matrix element, the amplitude integral of the COM wave function over the emission
region, so that the far-field radiative weight of eigenstate $n$ is
$\propto |\!\int\! \psi_n\, d^2\mathbf{r}|^2$ rather than the incoherent probability
density $|\psi_n|^2$; this is the standard result for disorder-localized exciton
emission in quantum wells and related nanostructures
\cite{runge1998spatially,runge1998level,zimmermann2003theory}.
For a local measurement that collects from an optical spot centered at $\mathbf{R}$
with amplitude collection mode $W_{\mathrm{amp}}$, the locally detected PL spectrum is
\begin{equation}
I(E,\mathbf{R})
= \sum_n f_{\mathrm{occ}}(E_n)\, A_n(\mathbf{R})\, g_\eta(E - E_n),
\label{eq:I_eigenstate_full}
\end{equation}
with the coherent radiative weight
\begin{equation}
A_n(\mathbf{R}) = |M_n(\mathbf{R})|^2, \;\;
M_n(\mathbf{R}) = \!\int\! W_{\mathrm{amp}}(\mathbf{r}{-}\mathbf{R})\, \psi_n(\mathbf{r})\, d^2\mathbf{r},
\label{eq:A_weight}
\end{equation}
where $W_{\mathrm{amp}}$ is the (real, positive) collection-mode amplitude and
$W_{\mathrm{opt}} = |W_{\mathrm{amp}}|^2$ is the intensity point-spread function.
For a diffraction-limited Gaussian spot both are Gaussian,
\begin{align}
W_{\mathrm{amp}}(\mathbf{r}) &= \frac{1}{\sqrt{2\pi\sigma_{\mathrm{opt}}^2}}\,
\exp\!\left(-\frac{|\mathbf{r}|^2}{4\sigma_{\mathrm{opt}}^2}\right),
\label{eq:W_amp_gaussian}
\\
W_{\mathrm{opt}}(\mathbf{r}) &= |W_{\mathrm{amp}}(\mathbf{r})|^2
= \frac{1}{2\pi\sigma_{\mathrm{opt}}^2}\,
\exp\!\left(-\frac{|\mathbf{r}|^2}{2\sigma_{\mathrm{opt}}^2}\right),
\label{eq:W_opt_gaussian}
\end{align}
so that $W_{\mathrm{amp}}$ has width $\sqrt{2}\,\sigma_{\mathrm{opt}}$ and the
normalized intensity PSF $W_{\mathrm{opt}}$ has width $\sigma_{\mathrm{opt}} =
\mathrm{FWHM}/(2\sqrt{2\ln 2}) \simeq 0.36~\mu$m (experimental FWHM $=0.85~\mu$m).
The occupation factor is $f_{\mathrm{occ}}(E)$. Only the coarse-graining scale
$\sigma_{\mathrm{opt}}$ enters the descriptor correlations; the overall
normalization sets only the absolute emission intensity.
The coherent weight \eqref{eq:A_weight} preferentially enhances nodeless localized
states and suppresses spatially extended states with sign-changing wave functions.
The mechanism is the cancellation in the amplitude
integral $M_n = \int W_{\mathrm{amp}}\,\psi_n$: a nodeless, well-localized trap
state keeps a single sign across its small extent, so its contributions add up
coherently and $A_n$ is large; an extended state whose wave function oscillates in
sign has its positive and negative contributions largely cancel, so $A_n$ is
small. Thus, in this approximation, localized trap-like states carry larger
radiative weights, whereas extended background states tend to be suppressed by
phase cancellation---a PL-level analogue of the mobility-edge selectivity known
for disordered quantum-well excitons
\cite{runge1998spatially,runge1998level}.

For the analytic descriptor theory it is convenient to use the incoherent
\emph{local-density approximation}, in which $|M_n(\mathbf{R})|^2$ is replaced by
$\int d^2\mathbf{r}\, W_{\mathrm{opt}}(\mathbf{r}-\mathbf{R})\,|\psi_n(\mathbf{r})|^2$,
so that Eq.~\eqref{eq:I_eigenstate_full} becomes the optically filtered LDOS
\begin{equation}
I(E,\mathbf{R})
= \int d^2\mathbf{r}\, W_{\mathrm{opt}}(\mathbf{r} - \mathbf{R})\,
f_{\mathrm{occ}}(E)\, \rho_\eta(E,\mathbf{r}).
\label{eq:I_formal}
\end{equation}
Because both weights are localized within the optical spot, the descriptor
covariance and the correlation-length hierarchy derived below are essentially
the same for the coherent and incoherent forms. We therefore use the transparent
LDOS form \eqref{eq:I_formal} both in the analytic treatment and in the
microscopic diagonalization of Sec.~\ref{sec:key_results}, and verify in
Appendix~\ref{app:optical_weight} that replacing it by the microscopically
motivated coherent radiative form \eqref{eq:A_weight} leaves the descriptor
statistics unchanged within the numerical seed-to-seed uncertainty (the two
weightings differ by $\le 0.01$ for every descriptor).

Approximating $f_{\mathrm{occ}} \approx 1$ for all states in the analyzed PL
window yields the working form
\begin{equation}
I(E,\mathbf{R})
= \sum_n A_n(\mathbf{R})\, g_\eta(E - E_n).
\label{eq:I_eigenstate}
\end{equation}
Here $A_n(\mathbf{R})$ denotes the local optical weight in either representation:
the coherent radiative form of Eq.~\eqref{eq:A_weight}, or the incoherent
local-density approximation
$A_n(\mathbf{R}) = \int d^2\mathbf{r}\, W_{\mathrm{opt}}(\mathbf{r}-\mathbf{R})\,|\psi_n(\mathbf{r})|^2$
[Eq.~\eqref{eq:I_formal}] used in the analytic treatment below.
A full treatment would fix $f_{\mathrm{occ}}$ from the relaxation kinetics between
disorder eigenstates \cite{runge1998level}, which at low temperature can be
strongly energy selective---for example through a mobility edge---and which we do
not attempt to model here. Instead, we adopt the constant-occupation form
$f_{\mathrm{occ}}\simeq 1$ as a deliberate minimal (null) model, \emph{not} as a
thermal-equilibrium approximation, so that the descriptor-correlation hierarchy is
attributed to the disorder-filtering mechanism rather than to occupation-induced
weighting. Within this minimal model, a spatially uniform, smooth energy-dependent
occupation would mainly reshape the common spectral envelope over the analyzed window
($\sim\!150$~meV, i.e.\ $1.25$--$1.40$~eV) and is not expected to change the spatial
correlation hierarchy, since the latter is set by the spatial statistics of the
eigenenergies $E_n$ and weights $A_n(\mathbf{R})$. A strongly energy-selective
occupation, or mobility-edge filtering of the kind emphasized in
Refs.~\cite{runge1998spatially,runge1998level}, would constitute an additional
physical weighting beyond the present minimal model, which we note here as a
limitation.
The optical spot intensity function is modeled as:

\begin{equation}
W_{\mathrm{opt}}(\mathbf{r})
= \frac{1}{2\pi\sigma_{\mathrm{opt}}^2}
\exp\!\left(-\frac{|\mathbf{r}|^2}{2\sigma_{\mathrm{opt}}^2}\right),
\label{eq:Wopt}
\end{equation}
with $\sigma_{\mathrm{opt}} = \mathrm{FWHM}/(2\sqrt{2\ln 2}) \approx 0.36~\mu$m
for a full width at half maximum (FWHM) of the optical spot of $0.85~\mu$m.

\subsection{Spectral Decomposition: Background and Trap Contributions}

A key analytical insight is to decompose $I(E,\mathbf{R})$
into contributions from states with different spatial extents.
We classify eigenstates using the inverse participation ratio (IPR),
which provides a standard measure of real-space localization
\cite{runge1998spatially} and thus naturally separates spatially extended
background states (small IPR) from trap-localized states (large IPR).
We define the localization length
$\ell_n = \mathrm{IPR}_n^{-1/2}$, where
$\mathrm{IPR}_n = \int d^2\mathbf{r}\, |\psi_n(\mathbf{r})|^4$
(with $\psi_n$ normalized).
The IPR works as an inverse effective area: for a state spread uniformly over an
area $A$ one has $|\psi_n|^2 \sim 1/A$ and $\mathrm{IPR}_n \sim 1/A$, so
$\ell_n = \mathrm{IPR}_n^{-1/2}$ measures the linear extent over which the state
has appreciable weight---small $\ell_n$ for tightly localized trap states, large
$\ell_n$ for extended background states \cite{runge1998spatially}.
A cutoff $\ell^*$ (of order the trap size) then separates the two families:

\begin{align}
I(E,\mathbf{R}) &= I_{\mathrm{bg}}(E,\mathbf{R}) + I_{\mathrm{sharp}}(E,\mathbf{R}),
\label{eq:I_decomp}
\\
I_{\mathrm{bg}}(E,\mathbf{R}) &= \sum_{n:\,\ell_n > \ell^*} A_n(\mathbf{R})\, g_\eta(E - E_n),
\label{eq:I_bg}
\\
I_{\mathrm{sharp}}(E,\mathbf{R}) &= \sum_{n:\,\ell_n \leq \ell^*} A_n(\mathbf{R})\, g_\eta(E - E_n),
\label{eq:I_sharp}
\end{align}
where $\ell^*$ is a crossover localization length (see Section~\ref{sec:localization}).
States with $\ell_n > \ell^*$ are spatially extended on the scale of the optical spot
and contribute to the smooth background envelope $I_{\mathrm{bg}}$.
Trap-concentrated states ($\ell_n \leq \ell^*$) contribute narrow lines to $I_{\mathrm{sharp}}$.
The IPR here serves to define this conceptual two-component decomposition, with
$\ell^*$ acting as a practical crossover rather than a sharp microscopic phase
boundary. In the strongly disordered regime of the numerical simulations, where
the IPR is set by the disorder landscape rather than by quantum interference, we
apply an equivalent energy-based classification of the same two families
(Section~\ref{sec:localization}).

\subsection{Optical Coarse-Graining of Slow Disorder}

The slow disorder component enters the measured PL spectrum primarily
through optical coarse-graining. The experimentally measured spectrum
at position $\mathbf{R}$ is not the LDOS at a single microscopic point,
but the LDOS averaged over the optical point-spread function
$W_{\mathrm{opt}}(\mathbf{r}-\mathbf{R})$ [Eq.~\eqref{eq:I_eigenstate}].
Since $V_{\mathrm{slow}}(\mathbf{r})$ varies on the length scale $\xi_s$,
much larger than the optical spot size $\sigma_{\mathrm{opt}}$
($\sigma_{\mathrm{opt}}/\xi_s \ll 1$), the smooth part of the local
spectral manifold within one optical spot is shifted almost uniformly
by the local average of $V_{\mathrm{slow}}$.

We therefore define the optically averaged slow potential as
\begin{equation}
\bar{V}_s(\mathbf{R})
= \int d^2\mathbf{r}\, W_{\mathrm{opt}}(\mathbf{r}-\mathbf{R})\, V_{\mathrm{slow}}(\mathbf{r}).
\label{eq:Born_result}
\end{equation}
This quantity is the energy shift of the background PL envelope observed
at the optical position $\mathbf{R}$, rather than the first-order energy
correction of a single microscopic eigenstate; it requires no assumption
about the detailed profile of individual wavefunctions.

Equivalently, for the background contribution to the LDOS one may write
\begin{equation}
\rho_{\mathrm{bg}}(E,\mathbf{R})
\simeq \rho_{\mathrm{bg}}^{(0)}\!\left(E - \bar{V}_s(\mathbf{R})\right),
\label{eq:bg_shift}
\end{equation}
up to corrections controlled by the small ratio $\sigma_{\mathrm{opt}}/\xi_s$.
This relation expresses the central coarse-graining assumption used below:
\emph{the smooth background within one optical spot is rigidly displaced by
$\bar{V}_s(\mathbf{R})$}, so that centroid-like descriptors, which integrate
spectral weight over the smooth envelope, primarily track $\bar{V}_s(\mathbf{R})$,
whereas peak-position descriptors remain sensitive to additional local
trap-level fluctuations within the same optical spot.

\section{Statistical Theory of Spectral Descriptors}
\label{sec:descriptor_theory}

This section makes the disorder-filter principle quantitative. We first define the
spectral descriptors---simple, peak-decomposition-free functionals of the local
spectrum---and then derive the statistics of the most important ones, showing
explicitly that the centroid energy tracks the slow disorder while the spectral
roughness and entropy track the traps. These results are the analytical basis for
the correlation-length hierarchy and the inter-descriptor correlations derived in
the following sections.

\subsection{Definitions of Spectral Descriptors}

We define the full set of spectral descriptors following the operational
conventions of the experiment~\cite{ahmad2025hierarchical}, so that simulated
and measured descriptors are directly comparable. Each descriptor is a simple
functional of the raw local spectrum $I(E,\mathbf{R})$ that requires no peak
fitting or microscopic line assignment. The set is deliberately chosen so that
different descriptors are sensitive to different aspects of the spectral
shape---its overall position, its width, its asymmetry about the dominant peak,
and its fine structure---and therefore, through the disorder-filter principle,
to different components of the underlying disorder: descriptors that integrate
the whole spectrum (such as the centroid) track the smooth background, whereas
peak- and roughness-based descriptors respond to the sharp, trap-generated fine
structure. We introduce them in turn, in each case stating what the descriptor
measures before giving its definition.

\subsubsection{Centroid Energy}

The centroid is the intensity-weighted mean emission energy---the first moment of
the local spectrum. Because it integrates over the entire spectrum, it averages
out the sub-spot trap fluctuations and tracks the smooth, slowly varying
background set by $V_{\mathrm{slow}}$; it is therefore the descriptor most
directly tied to the micron-scale disorder (Sec.~\ref{sec:hierarchy}). It is
defined by
\begin{equation}
E_{\mathrm{cent}}(\mathbf{R})
= \frac{\int E\, I(E,\mathbf{R})\, dE}{\int I(E,\mathbf{R})\, dE}.
\label{eq:Ecent}
\end{equation}

\subsubsection{Dominant-Peak Energy}

The dominant-peak energy is the position of the tallest spectral feature---the mode of
the local spectrum. In contrast to the centroid, it is set by whichever single
state carries the most optical weight within the optical spot, so it responds to
discrete trap-level switching as the spot is moved and, as shown below, acquires a
shorter correlation length than $E_{\mathrm{cent}}$. It is defined by
\begin{equation}
E_{\mathrm{dom}}(\mathbf{R})
= \operatorname*{arg\,max}_{E}\; I(E,\mathbf{R}).
\label{eq:Edom}
\end{equation}
Here, $\operatorname*{arg\,max}_{E}$ denotes the energy at which the
local PL spectrum reaches its maximum; for discretized spectra, it is
taken as the energy bin with the largest intensity.

\subsubsection{Centroid-Dominant Offset}

The offset between the centroid and the dominant peak measures how far the
intensity-weighted mean lies from the mode, i.e.\ the \emph{skewness} of the
emission envelope about its own maximum. It vanishes for a symmetric line and
takes a definite sign as spectral weight is redistributed to the low- or
high-energy side of the dominant peak, making it a sensitive, peak-decomposition-free
probe of envelope asymmetry. It is defined by
\begin{equation}
\Delta E_{cd}(\mathbf{R}) = E_{\mathrm{cent}}(\mathbf{R}) - E_{\mathrm{dom}}(\mathbf{R}).
\label{eq:DEcd}
\end{equation}

\subsubsection{Quantile Width}

The quantile width characterizes the overall energy spread of the emission
without reference to any individual peak. Using the central-$80\%$ band rather
than a variance or a fitted linewidth makes it robust to weak far tails and to
the multi-peak structure of the local spectrum, so that it reflects the
width of the dominant emission band. Specifically, $W_{80}(\mathbf{R})$ is the
energy width of the band containing the central $80\%$ of the integrated PL
intensity, i.e.\ the interval between the $10$th and $90$th
cumulative-intensity percentiles.

\subsubsection{High/Low Ratio}

The high/low ratio quantifies the asymmetry of the emitted intensity about the
dominant peak---how much spectral weight lies below the main line relative to
above it. It captures the same envelope skew as the offset $\Delta E_{cd}$, but
through the distribution of intensity rather than of energy, which is why the two
turn out to be so tightly linked (Sec.~\ref{sec:correlations}). Following the
experimental convention~\cite{ahmad2025hierarchical}, the spectrum is split at
the dominant peak $E_{\mathrm{dom}}$ (not at a fixed global energy):
\begin{equation}
R_{\mathrm{HL}}(\mathbf{R})
= \frac{\int_{-\infty}^{E_{\mathrm{dom}}} I(E,\mathbf{R})\, dE}
       {\int_{E_{\mathrm{dom}}}^{\infty} I(E,\mathbf{R})\, dE},
\label{eq:RHL}
\end{equation}
so that $R_{\mathrm{HL}} > 1$ when spectral weight accumulates below the dominant
emission peak. This shared $E_{\mathrm{dom}}$ reference is essential: it is what
makes $R_{\mathrm{HL}}$ and $\Delta E_{cd}$ two measures of the same envelope
asymmetry about the dominant peak, and hence tightly anticorrelated. Referencing
the split to $E_{\mathrm{cent}}$ or to a fixed spectral quantile instead measures
a different quantity and weakens or reverses the relation, as expected.

\subsubsection{Sharp Fraction}

The sharp fraction is the fraction of the emitted intensity carried by narrow,
trap-related lines rather than by the broad background envelope. It is a direct,
peak-decomposition-free proxy for the local trap loading---the number and depth
of sharp emitters within the optical spot---and is therefore the descriptor most
sensitive to the trap component of the disorder. It is defined by
\begin{equation}
F_S(\mathbf{R}) = \frac{I_{\mathrm{sharp, tot}}(\mathbf{R})}{I_{\mathrm{tot}}(\mathbf{R})},
\label{eq:FS}
\end{equation}
where $I_{\mathrm{sharp,tot}} = \int I_{\mathrm{sharp}}(E,\mathbf{R})\, dE$ is the
intensity carried by \emph{sharp} spectral features. Operationally, sharp peaks
are identified by local-maximum detection: a local maximum qualifies as sharp if
its prominence exceeds a fraction $p$ of the pixel maximum ($p = 0.02$) and its
width at half-prominence is below a threshold $w$ ($w = 10$~meV), which selects
narrow trap lines while excluding the broad background envelope
(FWHM $\sim 35$~meV); $I_{\mathrm{sharp}}$ is the intensity within $\pm w/2$ of
the qualifying peaks. This peak-detection criterion---rather than a fixed linewidth
threshold or a parametric background subtraction---is robust to its parameters:
varying $p$ over $0.01$--$0.05$ and $w$ over $8$--$15$~meV changes
$\rho_{\mathrm{S}}(F_S,R_1)$ by less than $0.03$ (from $+0.90$ to $+0.92$), while
the mean $\langle F_S\rangle$ shifts only within $0.25$--$0.39$.

\subsubsection{Spectral Roughness}

The spectral roughness is the total variation of the normalized spectrum---its
integrated absolute slope---and so quantifies the amount of sharp structure it
contains: a smooth envelope gives a small $R_1$, while each additional narrow
trap line adds to the total variation. It is thus a peak-detection-free proxy for
the density of fine structure, complementary to the sharp fraction $F_S$ (with
which it is strongly correlated, Sec.~\ref{sec:correlations}). Normalizing by the
integrated intensity makes it independent of the absolute emission strength:
\begin{equation}
R_1(\mathbf{R})
= \frac{\int |dI(E,\mathbf{R})/dE|\, dE}{\int I(E,\mathbf{R})\, dE}.
\label{eq:R1}
\end{equation}

\subsubsection{Spectral Entropy}

The spectral entropy measures how uniformly the emitted intensity is distributed
across energy, treating the normalized spectrum as a probability distribution. It
is small when the emission is concentrated in a few features and approaches unity
for a broad, feature-rich spectrum, and it therefore provides a
peak-decomposition-free measure of local spectral complexity that complements the
quantile width. Normalized by $\ln N$ to lie in $[0,1]$, it reads
\begin{equation}
S_{\mathrm{spec}}(\mathbf{R})
= -\frac{1}{\ln N} \sum_k p_k(\mathbf{R}) \ln p_k(\mathbf{R}),
\label{eq:Sspec}
\end{equation}
where $p_k = I(E_k,\mathbf{R}) / \sum_j I(E_j,\mathbf{R})$ is the
normalized spectral weight on pixel $k$ of the spectrometer,
and $N$ is the total number of energy pixels.

The fine-structure descriptors that involve a derivative or an energy-pixel sum
($R_1$, $S_{\mathrm{spec}}$) are not purely intrinsic quantities: their
\emph{absolute} values depend on the spectral resolution, the energy binning, any
smoothing applied before analysis, and the experimental noise floor. $R_1$
involves an energy derivative and so rescales under smoothing or spectral
broadening; $S_{\mathrm{spec}}$ depends on the effective number of spectral bins
and on the noise floor. We therefore do not treat these absolute values as
instrument-independent, and they should not be compared across datasets unless the
resolution and preprocessing are matched. In our simulations they are evaluated at
the experimental spectral resolution.

What we do compare is one level removed from the absolute values: the spatial
covariance of the descriptors, the inter-descriptor correlation signs, and the
correlation lengths, all computed with the \emph{same} preprocessing applied to
the experimental and simulated spectra. These quantities are expected to be more
robust to a uniform rebinning or broadening, because the same operation acts at
every spatial pixel while the pixel-to-pixel variation is set by the underlying
disorder landscape. We verified this explicitly only for $F_S$, whose
$F_S$--$R_1$ correlation shifts by less than $0.03$ under the threshold variations
above; we do not claim exact invariance for $R_1$ or $S_{\mathrm{spec}}$, and a
full propagation of instrumental noise and resolution through the covariance
matrix is left to future work. Accordingly, $R_1$, $F_S$, and
$S_{\mathrm{spec}}$ are best regarded as \emph{comparative} descriptors defined
together with a specified spectral-preprocessing protocol, rather than as absolute
instrument-independent observables; consistent preprocessing of experiment and
simulation is what makes the comparison meaningful.

\subsection{Centroid Energy: Slow Disorder Dominance}
\label{sec:Ecent_theory}

We now derive the statistical properties of each descriptor from the disorder model.

\begin{proposition}[Centroid energy tracks slow disorder]
\label{prop:Ecent}
When the integrated centroid is dominated by the smooth background envelope
rather than by the integrated weight of the sharp trap lines,
\begin{equation}
E_{\mathrm{cent}}(\mathbf{R}) \approx E_{\mathrm{cent}}^{(0)} + \bar{V}_s(\mathbf{R}),
\label{eq:Ecent_approx}
\end{equation}
where $E_{\mathrm{cent}}^{(0)}$ is the disorder-free centroid energy
and $\bar{V}_s(\mathbf{R})$ is the optical-spot-averaged slow disorder \eqref{eq:Born_result}.
\end{proposition}

\begin{proof}
From Eqs.\ \eqref{eq:I_eigenstate} and \eqref{eq:Ecent}:
\begin{equation}
E_{\mathrm{cent}}(\mathbf{R})
= \frac{\sum_n A_n(\mathbf{R})\, E_n}{\sum_n A_n(\mathbf{R})}.
\end{equation}
Using the rigid background displacement $E_n \approx E_n^{(0)} + \bar{V}_s(\mathbf{R}_n)$
from \eqref{eq:bg_shift},
and noting that $A_n(\mathbf{R})$ is concentrated near states $\psi_n$ whose support overlaps
the optical spot at $\mathbf{R}$, so $\bar{V}_s(\mathbf{R}_n) \approx \bar{V}_s(\mathbf{R})$
for all significantly contributing states:
\begin{equation}
E_{\mathrm{cent}}(\mathbf{R})
\approx \frac{\sum_n A_n(\mathbf{R})(E_n^{(0)} + \bar{V}_s(\mathbf{R}))}
              {\sum_n A_n(\mathbf{R})}
= E_{\mathrm{cent}}^{(0)} + \bar{V}_s(\mathbf{R}).
\end{equation}
\end{proof}

\begin{corollary}[Centroid correlation function]
\label{cor:Ecent_corr}
\begin{align}
C_{E_{\mathrm{cent}}}(r)
&\equiv \frac{\langle \delta E_{\mathrm{cent}}(\mathbf{R})\,
                     \delta E_{\mathrm{cent}}(\mathbf{R}') \rangle}
            {\langle \delta E_{\mathrm{cent}}^2 \rangle}
\nonumber\\
&\approx \exp\!\left(-\frac{r^2}{2(\xi_s^2 + 2\sigma_{\mathrm{opt}}^2)}\right),
\label{eq:CEcent}
\end{align}
for $r = |\mathbf{R} - \mathbf{R}'|$, where
$\delta E_{\mathrm{cent}} \equiv E_{\mathrm{cent}} - \langle E_{\mathrm{cent}}\rangle$.
\end{corollary}

\begin{proof}
$\bar{V}_s(\mathbf{R}) = \int d^2\mathbf{r}\, W_{\mathrm{opt}}(\mathbf{r}-\mathbf{R}) V_{\mathrm{slow}}(\mathbf{r})$
is a Gaussian-filtered version of $V_{\mathrm{slow}}$.
Setting $r = |\mathbf{R}-\mathbf{R}'|$, the convolution of two Gaussians gives:
\begin{align}
\langle \bar{V}_s(\mathbf{R}) \bar{V}_s(\mathbf{R}') \rangle
&= \int d^2\mathbf{r}\, d^2\mathbf{r}'\,
   W_{\mathrm{opt}}(\mathbf{r}-\mathbf{R})\,
   W_{\mathrm{opt}}(\mathbf{r}'-\mathbf{R}')\nonumber\\
&\quad \times W_s^2\, e^{-|\mathbf{r}-\mathbf{r}'|^2/2\xi_s^2}
\nonumber\\
&\simeq W_s^2 \exp\!\left(-\frac{r^2}{2(\xi_s^2 + 2\sigma_{\mathrm{opt}}^2)}\right),
\end{align}
since the convolution of Gaussians with widths $\sigma_{\mathrm{opt}}$, $\sigma_{\mathrm{opt}}$, $\xi_s$
yields a Gaussian with width $\sqrt{\xi_s^2 + 2\sigma_{\mathrm{opt}}^2}$.
The optical filtering reduces the prefactor to
$\Var(\bar{V}_s) = W_s^2\,\xi_s^2/(\xi_s^2 + 2\sigma_{\mathrm{opt}}^2)$;
for $\xi_s \gg \sigma_{\mathrm{opt}}$ this is $\simeq W_s^2$, as used above
(equivalently, one may read $W_s$ here as the filtered amplitude $W_{s,\mathrm{eff}} \simeq W_s$).
Since the normalized correlation \eqref{eq:CEcent} divides by the variance, this
prefactor cancels and Eq.~\eqref{eq:CEcent} is exact regardless of the reduction.
The effective correlation length is:
\begin{equation}
\xi_{\mathrm{eff}} = \sqrt{\xi_s^2 + 2\sigma_{\mathrm{opt}}^2}.
\label{eq:xi_eff}
\end{equation}
Since $\xi_s \gg \sigma_{\mathrm{opt}}$ ($2~\mu$m vs.\ $0.36~\mu$m), we have
$\xi_{\mathrm{eff}} \approx \xi_s$.
\end{proof}

\begin{remark}
The optical spot smoothing is negligible when $\xi_s \gg \sigma_{\mathrm{opt}}$,
which is satisfied experimentally.
The measured 1/e correlation length of $E_{\mathrm{cent}}$ directly yields $\xi_s$.
\end{remark}

\subsection{Spectral Entropy and Roughness: Trap Disorder Dominance}

In contrast to $E_{\mathrm{cent}}$, the spectral entropy and roughness
primarily reflect the local spectral manifold complexity generated by trap disorder.

\begin{proposition}[Entropy and roughness track trap density]
\label{prop:entropy}
For fixed slow disorder, the spectral entropy $S_{\mathrm{spec}}$ and
roughness $R_1$ are increasing functions of the effective trap density
$n_t^{\mathrm{eff}}(\mathbf{R}) = n_t \cdot f(\bar{V}_s(\mathbf{R}))$,
where $f$ is an increasing function of the slow disorder background
(deep traps attract excitons more effectively in a background minimum).
\end{proposition}

\textbf{Justification.} Each additional trap line resolved within the optical spot
adds total variation to $I(E,\mathbf{R})$, raising $R_1$, and adds a spectral
component, raising $S_{\mathrm{spec}}$ toward its bound; both therefore increase
with the number of trap lines per spot, i.e.\ with $n_t^{\mathrm{eff}}$. The
enhancement $f(\bar{V}_s)$ grows toward slow-potential minima, where deeper local
wells bind excitons more effectively. The argument is monotonic rather than exact
(near-degenerate lines contribute sub-additively), consistent with the
$\propto n_t W_t^2\sigma_t^2$ variance scaling of Appendix~\ref{app:variances}.

This proposition has an important corollary:
even the trap-sensitive descriptors $S_{\mathrm{spec}}$ and $R_1$
should show micron-scale spatial correlations,
because the effective trap coupling is modulated by $\bar{V}_s(\mathbf{R})$.
Specifically, regions with lower $E_{\mathrm{cent}}$ (deeper $V_{\mathrm{slow}}$ minima)
will show enhanced $F_S$ and $R_1$ because traps in such regions
create deeper bound states relative to the local band edge.

This gives the prediction:
\begin{equation}
\xi(S_{\mathrm{spec}}) \approx \xi(R_1) \approx \xi(F_S) \approx \xi_s,
\label{eq:xi_entropy}
\end{equation}
consistent with the experimental observation that all descriptors show
$\xi \approx 1$--$2~\mu$m.

\section{Correlation Length Hierarchy}
\label{sec:hierarchy}

This section derives the central analytical result of the paper: the
correlation-length hierarchy $\xi(E_{\mathrm{cent}}) \geq \xi(E_{\mathrm{dom}})$.
We state it as a proposition and prove it from the decomposition of the
dominant-peak energy into a slow part and a short-range trap-switching part. As
emphasized in the introduction, the informative content is not the mere ordering
(any average is smoother than a peak position) but its \emph{magnitude}, which
encodes the relative strength of the trap and slow disorder.

\subsection{Main Result}

\begin{proposition}[Correlation Length Hierarchy]
\label{prop:hierarchy}
For the model \eqref{eq:H}--\eqref{eq:Vtrap} with $W_s, W_t > 0$:
\begin{equation}
\xi(E_{\mathrm{cent}}) \geq \xi(E_{\mathrm{dom}}) \gtrsim \xi_t \equiv \sqrt{2}\sigma_t,
\label{eq:hierarchy_theorem}
\end{equation}
where $\xi_t$ is the trap correlation length.
The first inequality follows within the present decomposition from the
positivity of the short-range trap-switching variance
(Sec.~\ref{sec:hierarchy}); the second is a
microscopic floor that should be read as an order-of-magnitude lower bound.
In practice $\xi(E_{\mathrm{dom}})$ is not set by $\xi_t$ itself but by the
larger trap-switching scale $\ell_{\mathrm{switch}} \sim \sigma_{\mathrm{opt}}$
and is further broadened by optical and pixel ($\Delta x$) averaging, so the
effective floor is $\gtrsim \max(\xi_t, \sigma_{\mathrm{opt}}, \Delta x)$.
The equality $\xi(E_{\mathrm{cent}}) = \xi(E_{\mathrm{dom}})$ holds when the
trap-switching variance $\sigma_\delta^2$ vanishes, for example in the limit
$W_t \to 0$ (or when the dominant peak is always set by the smooth background,
so that traps never switch $E_{\mathrm{dom}}$).
\end{proposition}

\subsection{Proof}

\textbf{Step 1:} Decomposition of dominant-peak energy.
The dominant-peak energy $E_{\mathrm{dom}}$ is set by the position
of the highest spectral peak within the optical spot.
We write:
\begin{equation}
E_{\mathrm{dom}}(\mathbf{R})
= \bar{V}_s(\mathbf{R}) + \delta E_{\mathrm{dom}}(\mathbf{R}),
\label{eq:Edom_decomp}
\end{equation}
where $\delta E_{\mathrm{dom}}(\mathbf{R})$ captures the deviation of the dominant peak
from the smooth background.
Unlike the centroid, which averages over all states,
the dominant-peak energy is set by the single state with the highest optical weight
at position $\mathbf{R}$---which is typically a trap state or
a localized moir\'e state near $\mathbf{R}$.

\textbf{Step 2:} Statistics of $\delta E_{\mathrm{dom}}$.
The fluctuations $\delta E_{\mathrm{dom}}$ arise from:
(a) switching of trap levels within the optical spot as $\mathbf{R}$ varies,
(b) variation in which moir\'e-localized state has the highest weight.
Both sources produce fluctuations with correlation lengths
$\ell_{\mathrm{switch}} \sim \sigma_{\mathrm{opt}}$ (scale over which trap configuration changes
as the optical spot moves) and $\ell_t \sim \sigma_t$ (individual trap scale).
Since $\sigma_t \lesssim \sigma_{\mathrm{opt}} \ll \xi_s$:
\begin{equation}
\xi_{\delta E_{\mathrm{dom}}} \sim \sigma_{\mathrm{opt}} \ll \xi_s.
\end{equation}

\textbf{Step 3:} Correlation length of $E_{\mathrm{dom}}$.
The total $E_{\mathrm{dom}}$ correlation function is:
\begin{align}
\langle E_{\mathrm{dom}}(\mathbf{R}) E_{\mathrm{dom}}(\mathbf{R}') \rangle
&= \langle \bar{V}_s(\mathbf{R}) \bar{V}_s(\mathbf{R}') \rangle\nonumber\\
&\quad + \langle \delta E_{\mathrm{dom}}(\mathbf{R})
                \delta E_{\mathrm{dom}}(\mathbf{R}') \rangle
\nonumber\\
&= W_s^2 e^{-r^2/2\xi_{\mathrm{eff}}^2}
+ \sigma_{\delta}^2\, h(r/\sigma_{\mathrm{opt}}),
\label{eq:Edom_corr}
\end{align}
where $r = |\mathbf{R} - \mathbf{R}'|$ (as in Corollary~\ref{cor:Ecent_corr}),
$\sigma_\delta^2 = \Var(\delta E_{\mathrm{dom}})$, and $h(x) \to 0$ for $x \gg 1$.
In writing Eq.~\eqref{eq:Edom_corr} we have neglected the cross-correlation
$\langle \bar{V}_s(\mathbf{R})\,\delta E_{\mathrm{dom}}(\mathbf{R}')\rangle$.
Weak trap clustering correlates the trap positions with $V_{\mathrm{slow}}$ and
makes this term nonzero, but it is slaved to $\bar{V}_s$ and therefore only
renormalizes the effective short-range variance $\sigma_\delta^2$; it cannot
make $E_{\mathrm{dom}}$ smoother than $E_{\mathrm{cent}}$, so the ordering
$\xi(E_{\mathrm{cent}}) > \xi(E_{\mathrm{dom}})$ is preserved.
The variance satisfies:
\begin{equation}
\Var(E_{\mathrm{dom}}) = W_s^2 + \sigma_\delta^2
> W_s^2 = \Var(E_{\mathrm{cent}}),
\quad \sigma_\delta > 0.
\end{equation}
The normalized correlation function:
\begin{equation}
C_{E_{\mathrm{dom}}}(r)
= \frac{W_s^2 e^{-r^2/2\xi_{\mathrm{eff}}^2} + \sigma_\delta^2 h(r/\sigma_{\mathrm{opt}})}
       {W_s^2 + \sigma_\delta^2},
\end{equation}
decays from $1$ to zero at a scale smaller than $\xi_{\mathrm{eff}}$
because the $\sigma_\delta^2 h(r/\sigma_{\mathrm{opt}})$ term vanishes quickly
while the remaining $W_s^2 e^{-r^2/2\xi_{\mathrm{eff}}^2}/(W_s^2 + \sigma_\delta^2) < 1$
at $r = 0$.

At the $1/e$ scale of the experimentally relevant hierarchy,
$\xi(E_{\mathrm{dom}}) \gtrsim 1~\mu$m $\gg \sigma_{\mathrm{opt}}$, so the
short-range term $\sigma_\delta^2 h(r/\sigma_{\mathrm{opt}})$ has already decayed
and only the smooth background contribution survives. The 1/e length
$\xi(E_{\mathrm{dom}})$ then satisfies the implicit equation
\begin{equation}
\frac{W_s^2}{W_s^2 + \sigma_\delta^2} e^{-\xi(E_{\mathrm{dom}})^2/2\xi_s^2} = e^{-1}.
\label{eq:xi_Edom_implicit}
\end{equation}
Solving Eq.~\eqref{eq:xi_Edom_implicit} in the perturbative regime
$\sigma_\delta^2/W_s^2 < 1$ gives a shortened correlation length
$\xi(E_{\mathrm{dom}}) < \xi_s \approx \xi(E_{\mathrm{cent}})$; in the
weak-trap-noise limit $\sigma_\delta^2 \ll W_s^2$ this reduces to the
closed-form expansion
\begin{equation}
\xi(E_{\mathrm{dom}}) \approx \xi_s\left(1 - \frac{\sigma_\delta^2}{2W_s^2}\right).
\label{eq:xi_Edom_weak}
\end{equation}
Outside the perturbative regime, $\xi(E_{\mathrm{dom}})$ must be obtained
directly from Eq.~\eqref{eq:Edom_corr}; only the qualitative ordering
$\xi(E_{\mathrm{dom}}) < \xi(E_{\mathrm{cent}})$ is preserved within the present
decomposition.

\textbf{Order-of-magnitude estimate.}
From the experiment~\cite{ahmad2025hierarchical}: $\xi(E_{\mathrm{cent}}) \approx 2.00~\mu$m,
$\xi(E_{\mathrm{dom}}) \approx 1.27~\mu$m, so
$\xi(E_{\mathrm{dom}})/\xi(E_{\mathrm{cent}}) \approx 0.64$.
Inserting this ratio into the Gaussian implicit
relation~\eqref{eq:xi_Edom_implicit} (with $\xi(E_{\mathrm{cent}}) \approx \xi_s$)
and solving for the trap-switching variance gives
\begin{equation}
\frac{\sigma_\delta^2}{W_s^2}
= \exp\!\left[\,1 - \frac{\xi(E_{\mathrm{dom}})^2}{2\xi_s^2}\,\right] - 1
\approx 1.2,
\label{eq:sigma_delta_estimate}
\end{equation}
i.e., a trap-switching fluctuation $\sigma_\delta \sim W_s$ of order unity.
The independent estimate from the centroid--dominant correlation
[Eq.~\eqref{eq:rho_centdom} below, $\sigma_\delta/W_s \approx 0.8$]
falls in the same range; the residual spread reflects the assumed
functional form of the mixed correlation decay, so we read this only as an
order-of-magnitude estimate. Both routes place the trap disorder comparable to
the slow disorder amplitude, in the hierarchical regime, and the simulated
correlation functions [Fig.~\ref{fig:correlation_hierarchy}(a)] confirm the
hierarchy directly, as summarized by the correlation-length bar chart
[Fig.~\ref{fig:correlation_hierarchy}(b)].

\textbf{Range of validity.}
The weak-noise expansion Eq.~\eqref{eq:xi_Edom_weak} is strictly controlled
only for $\sigma_\delta^2/W_s^2 \ll 1$. The experimentally inferred ratio
$\sigma_\delta^2/W_s^2 \approx 1.2$ lies outside this perturbative regime,
so the numerical inversion above should be read as an
\emph{order-of-magnitude estimate} based on the full two-component
correlation form Eq.~\eqref{eq:Edom_corr}, not as a small-noise expansion.
The sign of the hierarchy
$\xi(E_{\mathrm{dom}}) < \xi(E_{\mathrm{cent}})$ is robust within the present
decomposition for any $\sigma_\delta > 0$, following from the positivity of the
trap-switching variance; only the precise mapping
$\sigma_\delta^2/W_s^2 \leftrightarrow \xi(E_{\mathrm{dom}})/\xi_s$
is approximation-controlled.

\begin{figure*}[t]
\centering
\includegraphics[width=\textwidth]{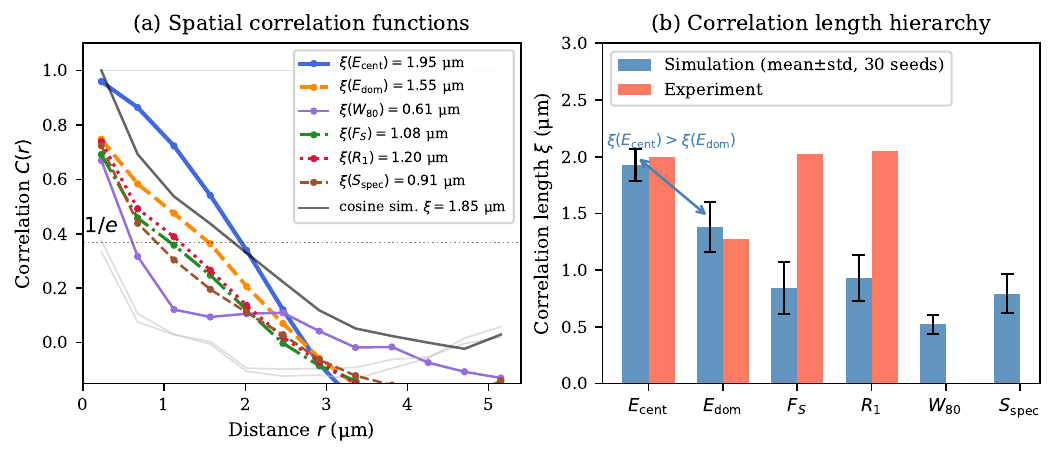}
\caption{
\textbf{Spatial correlation length hierarchy confirmed by simulation.}
(a) Normalized spatial correlation functions $C(r)$ for six spectral descriptors,
computed from a single $20\times 20$ hyperspectral map on an $8\times 8~\mu$m$^2$
domain with best-fit parameters ($W_s=6~$meV, $\xi_s=2~\mu$m, $W_t=12~$meV).
The horizontal dotted line marks the $1/e$ threshold, and the legend lists the
$1/e$ length of each highlighted descriptor for this single realization
(the ensemble means over thirty realizations are shown in panel (b)).
Thin light-gray lines are the remaining, un-highlighted descriptors
($\Delta E_{cd}$, $R_{\mathrm{HL}}$, $I_{\mathrm{tot}}$).
The dark-gray solid curve is the whole-spectrum cosine similarity. Because every
pixel shares the same broad envelope, its raw value sits on a high nonzero
plateau; we therefore subtract that long-range plateau and renormalize to unity
at $r=0$, so that only the physically decaying excess is shown. Its $1/e$ length,
$1.85~\mu$m, matches the experimental whole-spectrum value of $1.87~\mu$m.
(b) Simulated (blue) versus experimental (red) $1/e$ correlation lengths. The
simulated bars are the mean $\pm$ standard deviation over thirty disorder
realizations, with $\xi$ extracted as the $1/e$ decay length of the spatial
autocorrelation (the experimental convention). The simulation confirms the
hierarchy $\xi(E_{\mathrm{cent}}) > \xi(E_{\mathrm{dom}})$ (blue arrow):
$\xi(E_{\mathrm{cent}}) = 1.93\pm0.14~\mu$m and
$\xi(E_{\mathrm{dom}}) = 1.38\pm0.22~\mu$m (ratio $0.72\pm0.10$), consistent with
the experimental $2.00$ and $1.27$ (ratio $0.64$)~\cite{ahmad2025hierarchical}.
The $F_S$ and $R_1$ lengths fall below the experimental values, reflecting the
Poisson-limited trap statistics of the simplified model
(Sec.~\ref{sec:discussion}); their finite-size robustness is examined in
Appendix~\ref{app:finite_size}.
}
\label{fig:correlation_hierarchy}
\end{figure*}

\section{Inter-Descriptor Correlation Theory}
\label{sec:correlations}

By an \emph{inter-descriptor correlation} we mean the correlation between two
spatial descriptor maps $X(\mathbf{R})$ and $Y(\mathbf{R})$ evaluated over all
measured pixels $\mathbf{R}$ of the hyperspectral map; it thus measures how two
coarse descriptors of the local spectrum co-vary across the sample, not a
correlation between individual spectral peaks.
Throughout this section, correlations between two such descriptor fields
$X(\mathbf{R})$ and $Y(\mathbf{R})$ are quantified by the Pearson coefficient
\begin{equation}
\rho_{\mathrm{P}}(X,Y) = \frac{\Cov(X,Y)}{\sqrt{\Var(X)\,\Var(Y)}},
\label{eq:pearson_def}
\end{equation}
where $\Cov(X,Y) = \langle XY\rangle - \langle X\rangle\langle Y\rangle$ and
$\Var(X) = \Cov(X,X)$, with $\langle\cdot\rangle$ denoting the spatial average over
the map. The experimental descriptor matrix, by contrast, is reported using the
Spearman rank-correlation coefficient $\rho_{\mathrm{S}}$~\cite{ahmad2025hierarchical},
defined as the Pearson correlation of the rank-transformed descriptor fields,
\begin{equation}
\rho_{\mathrm{S}}(X,Y) = \rho_{\mathrm{P}}(\mathrm{rank}\,X,\ \mathrm{rank}\,Y),
\label{eq:spearman_def}
\end{equation}
where $\mathrm{rank}\,X$ denotes the rank-transformed values of $X(\mathbf{R})$
over all map pixels. It thus measures monotonic association without assuming a
linear relation. The Spearman coefficient is the
natural choice here because the spectral descriptors are nonlinear functionals of
the local PL spectrum and need not be Gaussian distributed; rank correlations
therefore provide a robust measure of monotonic co-variation that is directly
comparable to the experimental descriptor-correlation matrix. We compute
$\rho_{\mathrm{P}}$
analytically because it is exact for the Gaussian disorder variables underlying our model,
and use it as a sign-and-magnitude estimate for $\rho_{\mathrm{S}}$: for the monotonic
descriptor relations considered here, the two coefficients share the same sign and have
comparable magnitude, so that the analytical $\rho_{\mathrm{P}}$ predicts the measured
$\rho_{\mathrm{S}}$. Where we quote our own simulation values for comparison with experiment
(e.g.\ Sec.~\ref{sec:key_results}), we report $\rho_{\mathrm{S}}$ to match the experimental
convention.

\subsection{Centroid-Dominant Correlation}

From Eqs.\ \eqref{eq:Ecent_approx} and \eqref{eq:Edom_decomp}, the Pearson correlation
between the centroid and peak-position fields is
\begin{align}
\rho_{\mathrm{P}}(E_{\mathrm{cent}}, E_{\mathrm{dom}})
&= \frac{\Cov(\bar{V}_s,\, \bar{V}_s + \delta E_{\mathrm{dom}})}
       {\sqrt{\Var(\bar{V}_s)}\,\sqrt{\Var(\bar{V}_s) + \Var(\delta E_{\mathrm{dom}})}}
\nonumber\\
&= \frac{W_s}{\sqrt{W_s^2 + \sigma_\delta^2}} < 1.
\label{eq:rho_centdom}
\end{align}
Here $W_s$ denotes the optically filtered slow-disorder amplitude
$\Var(\bar{V}_s)^{1/2}$ (Sec.~\ref{sec:spectral_theory}), which is
indistinguishable from the bare $W_s$ in the limit $\xi_s \gg \sigma_{\mathrm{opt}}$.
Identifying this with the experimental Spearman value~\cite{ahmad2025hierarchical} $\rho_{\mathrm{S}} \approx 0.79$
gives $\sigma_\delta/W_s \approx 0.8$, of the same order as the correlation-length estimate above
($\sigma_\delta/W_s \sim 1$); both indicate $\sigma_\delta \sim W_s$.

\subsection{The $\Delta E_{cd}$--$R_{\mathrm{HL}}$ Anti-Correlation}

The strong experimental anti-correlation $\rho_{\mathrm{S}}(\Delta E_{cd}, R_{\mathrm{HL}}) \approx -0.978$~\cite{ahmad2025hierarchical}
is the most striking inter-descriptor correlation.
We now derive this as a robust spectral shape relation for spectra with a dominant unimodal envelope.

\begin{proposition}[$\Delta E_{cd}$--$R_{\mathrm{HL}}$ anti-correlation]
\label{prop:DEcd_RHL}
Consider spectra with a dominant unimodal envelope whose shape is governed by a
single asymmetry parameter $\alpha$, with $\alpha=0$ the case symmetric about the
dominant peak $E_{\mathrm{dom}}$. The centroid-dominant offset
$\Delta E_{cd} = E_{\mathrm{cent}} - E_{\mathrm{dom}}$ and the high/low ratio
$R_{\mathrm{HL}}$ are then monotone in $\alpha$ with opposite sense:
$\Delta E_{cd}$ and $R_{\mathrm{HL}}-1$ vanish together at $\alpha=0$ and carry
opposite signs for $\alpha\neq 0$,
\begin{align*}
\Delta E_{cd} < 0 &\iff R_{\mathrm{HL}} > 1, \\
\Delta E_{cd} > 0 &\iff R_{\mathrm{HL}} < 1,
\end{align*}
so that $\rho_{\mathrm{P}}(\Delta E_{cd}, R_{\mathrm{HL}}) \to -1$ when $\alpha$ is
the dominant source of pixel-to-pixel variation. (Absent the single-parameter
restriction the sign relation is a strong tendency rather than an identity, since
$\Delta E_{cd}$ weights the below/above intensities by their mean energy distances
whereas $R_{\mathrm{HL}}$ weights them equally.)
\end{proposition}

\begin{proof}
The centroid $E_{\mathrm{cent}} = \int E I(E) dE / \int I(E) dE$
weights high-energy spectral weight more heavily relative to $E_{\mathrm{dom}}$
when the spectrum has a high-energy tail.
A low-energy tail concentrates weight below $E_{\mathrm{dom}}$, giving $R_{\mathrm{HL}} > 1$
and pulling the centroid down so that $E_{\mathrm{cent}} < E_{\mathrm{dom}}$ and $\Delta E_{cd} < 0$.
Conversely, a high-energy tail gives $R_{\mathrm{HL}} < 1$ and $\Delta E_{cd} > 0$.

More precisely, define $m_1^{>} = \int_{E_{\mathrm{dom}}}^{\infty} E I(E) dE$
and $m_1^{<} = \int_{-\infty}^{E_{\mathrm{dom}}} E I(E) dE$.
Then:
\begin{align}
E_{\mathrm{cent}} - E_{\mathrm{dom}}
&= \frac{m_1^{>} + m_1^{<}}{I_{\mathrm{tot}}} - E_{\mathrm{dom}}
\nonumber\\
&= \frac{m_1^{>} - E_{\mathrm{dom}} I_{\mathrm{tot}}^{>}}{I_{\mathrm{tot}}}
  + \frac{m_1^{<} - E_{\mathrm{dom}} I_{\mathrm{tot}}^{<}}{I_{\mathrm{tot}}},
\end{align}
where $I_{\mathrm{tot}}^{>/<}$ are the integrated intensities above/below $E_{\mathrm{dom}}$.
Both terms represent signed ``moments'' of spectral weight about $E_{\mathrm{dom}}$,
which are monotone functions of spectral asymmetry.

The ratio $R_{\mathrm{HL}} = I_{\mathrm{tot}}^{<}/I_{\mathrm{tot}}^{>}$
is a monotone function of the same asymmetry parameter $\alpha$.
Since both $\Delta E_{cd}$ and $R_{\mathrm{HL}}-1$ are monotone in $\alpha$ and
vanish at $\alpha=0$ with opposite sign, their joint locus is a monotone curve of
negative slope, giving $\rho_{\mathrm{P}}(\Delta E_{cd}, R_{\mathrm{HL}}) \to -1$
in the limit of pure asymmetry variation.
This argument refers to the asymmetry of the broad envelope about its dominant
maximum, not to the placement of individual lines: an isolated low-energy trap
peak can itself set $E_{\mathrm{dom}}$ and produce $\Delta E_{cd} > 0$ (as in
the trap-rich pixels of Fig.~\ref{fig:spectra_maps}), but the observed
near-one-dimensional $\Delta E_{cd}$--$R_{\mathrm{HL}}$ relation indicates that
it is the envelope asymmetry, rather than such discrete level switching, that
remains the controlling variable across the map.
Because both descriptors are referenced to the dominant peak, the
anticorrelation requires this shared reference: in our simulated maps
$\rho_{\mathrm{S}}(\Delta E_{cd}, R_{\mathrm{HL}}) = -0.97\pm0.01$ with the
$E_{\mathrm{dom}}$ split, but only $+0.50$ with an $E_{\mathrm{cent}}$ split and
$+0.11$ with a fixed-median split---showing that the strong relation is tied to
the dominant-peak-referenced definition (the experimentally adopted one), rather
than a consequence of an arbitrary normalization choice.
\end{proof}

\begin{remark}
The near-perfect anti-correlation $\rho_{\mathrm{S}} \approx -0.978$ observed experimentally
indicates that spectral asymmetry variation---rather than peak multiplicity or
peak energy modulation---is the dominant source of descriptor variation in this dataset.
This is a nontrivial diagnostic: the mere observation that
$\rho_{\mathrm{S}}(\Delta E_{cd}, R_{\mathrm{HL}}) \approx -1$ is difficult to reconcile with simple models in which
multiple incoherent peaks at unrelated energies dominate the spectrum,
and instead establishes that the broad envelope remains effectively unimodal
across the map---directly confirming the two-fluid picture of Sec.~\ref{sec:localization}.
\end{remark}

\subsection{Sharp Fraction -- Roughness Correlation}

Both $F_S$ and $R_1$ measure the density of sharp spectral structure.
They differ in how the structure is weighted:
$F_S$ measures the fraction of integrated intensity in narrow features,
while $R_1$ measures the total variation of the spectral profile.
For a spectrum consisting of a smooth background plus sharp peaks,
both increase with the number and depth of sharp peaks,
giving a large positive correlation.

Analytically, for a spectrum $I = I_{\mathrm{bg}} + \sum_i A_i g_\eta(E - E_i)$:
\begin{align}
R_1 &\approx R_1^{\mathrm{bg}} + \frac{2}{\eta\sqrt{2\pi}} \cdot \frac{\sum_i A_i}{I_{\mathrm{tot}}},
\\
F_S &\approx \frac{\sum_i A_i}{I_{\mathrm{tot}}},
\end{align}
giving $R_1 \approx R_1^{\mathrm{bg}} + \mathrm{const} \times F_S$,
and thus $\rho_{\mathrm{P}}(F_S, R_1) \to 1$ when trap peak variation dominates.
The experimental value $\rho_{\mathrm{S}} \approx 0.901$~\cite{ahmad2025hierarchical} confirms this.

\subsection{Width -- Entropy Correlation}

The quantile width $W_{80}$ and spectral entropy $S_{\mathrm{spec}}$ both measure
the spread of spectral weight.
For a spectrum spanning $N_{\mathrm{eff}}$ effective resolution elements:
$W_{80} \propto N_{\mathrm{eff}}$ and $S_{\mathrm{spec}} \propto \ln N_{\mathrm{eff}}$.
Since both are monotone functions of $N_{\mathrm{eff}}$, they are positively correlated.
The experimental value $\rho_{\mathrm{S}} \approx 0.846$~\cite{ahmad2025hierarchical} is consistent.
The full simulated Spearman inter-descriptor matrix is shown in
Fig.~\ref{fig:spearman_matrix}.

\begin{figure*}[t]
\centering
\includegraphics[width=0.55\textwidth]{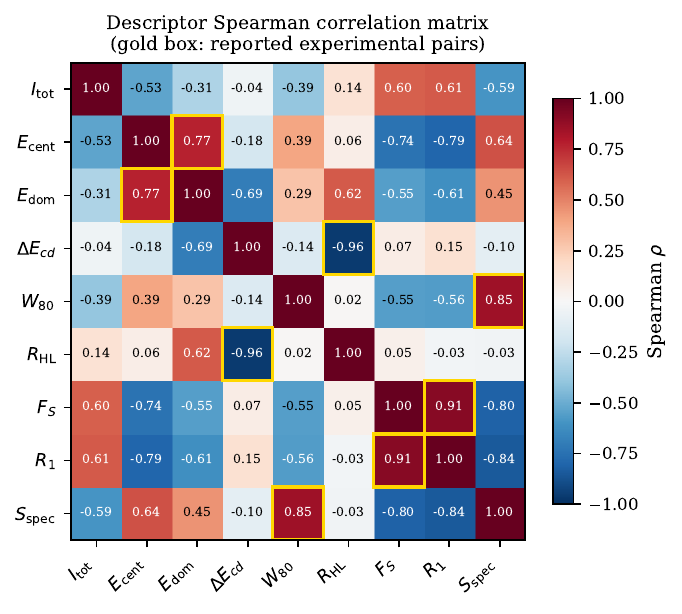}
\caption{
\textbf{Spearman inter-descriptor correlation matrix from simulation.}
Each cell shows the Spearman rank correlation coefficient $\rho_{\mathrm{S}}$ between
a pair of spectral descriptors, computed from a simulated $20\times 20$
hyperspectral map with best-fit parameters.
Gold boxes mark the four pairs also reported experimentally~\cite{ahmad2025hierarchical}---%
$\Delta E_{cd}$--$R_{\mathrm{HL}}$, $F_S$--$R_1$, $W_{80}$--$S_{\mathrm{spec}}$,
and $E_{\mathrm{cent}}$--$E_{\mathrm{dom}}$.
The displayed coefficients are those of a single best-fit map; the five-seed mean
simulated values and their measured counterparts are compared in
Table~\ref{tab:spearman} (mean absolute deviation $0.010$ per pair).
Note the strong positive block structure among
$(I_{\mathrm{tot}}, F_S, R_1)$ (trap-rich pixels) and
the negative block coupling these to $(E_{\mathrm{cent}}, E_{\mathrm{dom}})$,
reflecting the low-energy trap emission at potential minima.
}
\label{fig:spearman_matrix}
\end{figure*}

\section{Disorder Parameter Regimes}
\label{sec:phase_diagram}

The qualitative behavior of the model is set by the two disorder strengths---the
slow-disorder amplitude $W_s$ and the effective trap disorder $\Gamma_t$---measured
against the homogeneous linewidth $\eta$. Which of these exceeds $\eta$ determines
whether the local spectra are essentially featureless, smoothly modulated on the
micron scale, densely peaked, or organized on both scales at once. In this section
we map that dependence onto four qualitatively distinct regimes, in the spirit of a
phase diagram (Fig.~\ref{fig:phase_diagram}), and use the map to locate the
experimental system. This matters because the descriptor correlation-length
hierarchy studied here is prominent only when \emph{both} disorder scales exceed
$\eta$; identifying the regime therefore tells us when the descriptor analysis
applies. We first introduce the two control parameters, then define the regime
boundaries, and finally place the MoSe$_2$/WSe$_2$ system on the resulting map.

\subsection{Control Parameters}

We map the parameter space of the two disorder strengths into qualitatively distinct regimes:
\begin{align}
&W_s \quad \text{(slow disorder amplitude, meV)},
\nonumber\\
&\Gamma_t \equiv n_t W_t^2 \sigma_t^2
\quad \text{(effective trap disorder, meV}^2).
\label{eq:control_params}
\end{align}
The quantity $\Gamma_t$ is the effective trap-potential variance,
combining the trap density $n_t$, the depth variance $W_t^2$, and the trap area
$\sim\sigma_t^2$ through the dimensionless coverage $n_t\sigma_t^2$; its square
root $\Gamma_t^{1/2}$ is the rms trap-potential amplitude (in meV), directly
comparable to $W_s$ and $\eta$.

\subsection{Regime Boundaries}

We define four regimes based on the hierarchy of spectral disorder.
The boundaries below are order-of-magnitude criteria intended to
demarcate qualitatively distinct regimes; the exact boundaries depend
on microscopic details and should be regarded as schematic.
The geometric factors $(a_M/\xi_s)^{1/2}$ that appear below are schematic
coarse-graining factors estimating how nanoscale trap fluctuations (on the
moir\'e scale $a_M$) survive averaging over a slow-disorder domain of size
$\xi_s$; they are heuristic, not the result of a microscopic derivation.

\paragraph{Regime~I: Uniform Excitonic (small $W_s$, small $\Gamma_t$).}
Neither disorder scale produces significant spectral structure.
The PL is a single smooth Gaussian-like peak at the unperturbed exciton energy.
Boundary: $W_s < \eta$ (spectral resolution) and $\Gamma_t^{1/2} < \eta$.

\paragraph{Regime~II: Smooth Correlated-Domain (large $W_s$, small $\Gamma_t$).}
Slow disorder creates micron-scale spectral domains.
The spectrum at each pixel is smooth (low roughness, low entropy, small $F_S$).
$E_{\mathrm{cent}}$ maps show smooth spatial modulation with $\xi(E_\mathrm{cent}) \approx \xi_s$.
Boundary: $W_s > \eta$ and $\Gamma_t^{1/2} < W_s \cdot (a_M/\xi_s)^{1/2}$.

\paragraph{Regime~III: Random Line-Rich (small $W_s$, large $\Gamma_t$).}
Dense sharp spectral peaks appear at random positions.
High roughness and entropy, but no spatial domain structure.
Boundary: $W_s < \eta$ and $\Gamma_t^{1/2} > \eta$.

\paragraph{Regime~IV: Hierarchical Disorder-Dominated (large $W_s$, large $\Gamma_t$).}
Both disorder scales are active.
Micron-scale spectral domains coexist with local multipeak manifolds.
The descriptor correlation-length hierarchy $\xi(E_\mathrm{cent}) > \xi(E_\mathrm{dom})$
is maximally observable and directly measurable in this regime.

\begin{align}
\text{Regime~IV:}\quad
W_s &> \eta, \nonumber\\
\Gamma_t^{1/2} &> W_s \cdot (a_M/\xi_s)^{1/2}.
\label{eq:phase4}
\end{align}

\subsection{Location of the Experimental System}

The experimental data directly constrain the model parameters. The measured $\xi(E_{\mathrm{cent}}) \approx 2.00~\mu$m implies $\xi_s \approx 2~\mu$m; the centroid energy variation across the map implies $W_s \sim 10$--$15~$meV; and the prevalence of multipeak spectra with $F_S > 0.1$ and $R_1 \sim 0.3$--$0.5~$meV$^{-1}$ confirms $\Gamma_t^{1/2} > \eta \sim 2~$meV.
The experimental system therefore falls within Regime~IV (hierarchical disorder-dominated),
as illustrated in the disorder parameter map (Fig.~\ref{fig:phase_diagram}).
The value $W_s \sim 10$--$15~$meV here is the \emph{raw centroid variation}
across the experimental map. In the numerical simulations of
Sec.~\ref{sec:numerics} we use $W_s = 6~$meV as an \emph{effective}
slow-disorder amplitude that accounts for optical-spot averaging and
finite-map normalization, so that the two values refer to different
operational definitions of the slow-disorder strength rather than a
quantitative inconsistency.

\begin{figure}[ht]
\centering
\includegraphics[width=\columnwidth]{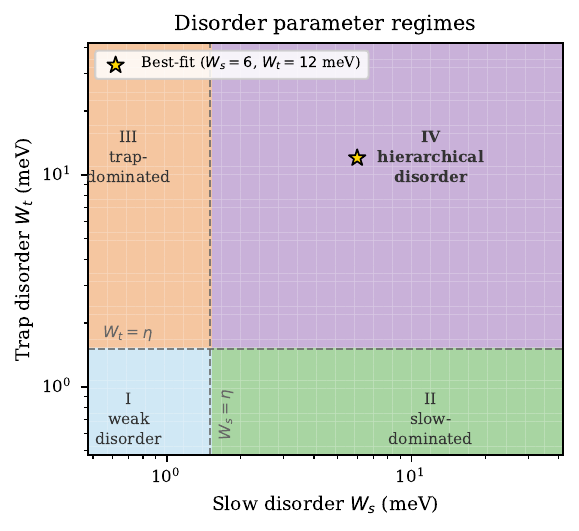}
\caption{
\textbf{Disorder parameter map for the hierarchical disorder model.}
Four qualitatively distinct regimes are defined by the disorder strengths relative to
the single-level linewidth $\eta = 1.5~$meV:
Regime~I (weak disorder, $W_s, W_t < \eta$): spectrally uniform emission;
Regime~II (slow-dominated, $W_s > \eta > W_t$): smooth micron-scale spectral modulation;
Regime~III (trap-dominated, $W_t > \eta > W_s$): dense random sharp lines without spatial organization;
Regime~IV (hierarchical disorder-dominated, $W_s, W_t > \eta$): coexisting micron-scale domains
and dense local trap manifolds; the descriptor correlation-length hierarchy is maximally visible.
The gold star marks the best-fit parameters reproducing the four principal
Spearman descriptor correlations reported experimentally~\cite{ahmad2025hierarchical}
simultaneously.
Here $n_t$ and $\sigma_t$ are held fixed, so the vertical axis $W_t$ represents
the trap-disorder strength $\Gamma_t^{1/2} = (n_t W_t^2 \sigma_t^2)^{1/2}$ of
Eq.~\eqref{eq:control_params} up to a constant factor.
}
\label{fig:phase_diagram}
\end{figure}

\section{Two-Fluid LDOS Structure}
\label{sec:localization}

The eigenstates of $H = -(\hbar^2/2M)\nabla^2 + V_{\rm slow} + V_{\rm trap}$
fall into two populations based on their energy relative to the local slow potential.

\textbf{Background states} ($E_n \gtrsim \bar{V}_s(\mathbf{R}_n)$) are weakly bound
to slow-disorder minima, with spatial extents $\ell_n \gg \sigma_t$.
Each contributes a broad Lorentzian of width $\eta_{\rm bg}$ centered near
$\bar{V}_s(\mathbf{R}_n)$.
The incoherent sum of $\sim$20 background eigenstates per optical spot
produces a smooth, quasi-Gaussian envelope that tracks $\bar{V}_s(\mathbf{R})$
and constitutes the broad background $I_{\rm bg}$.

\textbf{Trap states} (those lying more than ${\sim}\,W_t$ below the local
background, $E_n \lesssim \bar{V}_s(\mathbf{R}_n) - W_t$)
are spatially concentrated near individual trap sites $\mathbf{R}_i$,
with spatial extent $\ell_n \sim \sigma_t$.
Each contributes a narrow peak of width $\eta \ll \eta_{\rm bg}$
at energy $E_n \approx \bar{V}_s(\mathbf{R}_i) + E_n^{(\rm trap)}$.

The resulting local density of states at position $\mathbf{R}$ is:
\begin{align}
\rho(E,\mathbf{R})
&\approx \underbrace{\rho_{\mathrm{bg}}(E - \bar{V}_s(\mathbf{R}))}_{\text{smooth background}}
\nonumber\\
&\quad + \underbrace{
  \sum_{\substack{i:\,\mathbf{R}_i \in \mathrm{spot}(\mathbf{R})}}
  \rho_i^{\mathrm{trap}}(E - \bar{V}_s(\mathbf{R}_i))
}_{\text{sharp trap peaks}},
\label{eq:LDOS_twofluid}
\end{align}
and the local PL spectrum follows from integrating $\rho(E,\mathbf{R})$ against
the optical weight function.
This two-fluid structure is reproduced at the model level by the Hamiltonian diagonalization
in Fig.~\ref{fig:emergence}(b): the filled areas show the background and trap
contributions separately.

A key non-trivial consequence is that the broad-background intensity $I_{\rm bg,tot}$
and sharp-peak intensity $I_{\rm sharp,tot}$ are \emph{positively} correlated:
\begin{equation}
\Cov(I_{\mathrm{bg,tot}}, I_{\mathrm{sharp,tot}}) > 0.
\label{eq:covariance_positive}
\end{equation}
This follows because both populations are enhanced in slow-potential minima:
background eigenstates are preferentially concentrated in slow-potential minima
by disorder selection, while traps are preferentially nucleated at the same sites.
The positive covariance distinguishes the hierarchical disorder landscape
from a two-phase competition model, in which sharp peaks would form
at the expense of the background.

\section{Spectral Families as Correlated Disorder Domains}
\label{sec:families}

\subsection{Origin of Spectral Families}

The experimental observation of three dominant spectral families
in PCA + Gaussian mixture clustering
raises the question: are these families distinct thermodynamic phases,
or emergent statistical features of the disorder landscape?

In our framework, spectral families emerge naturally from the
correlated slow disorder field $V_{\mathrm{slow}}(\mathbf{r})$.
A Gaussian random field on a finite domain $[0,L]^2$ with correlation length
$\xi_s \ll L$ develops $N_{\mathrm{dom}} \sim (L/\xi_s)^2$ quasi-independent
regions with characteristic local values $V_{\mathrm{slow}} \sim W_s$.

For the experiment~\cite{ahmad2025hierarchical}: $L \approx 8~\mu$m, $\xi_s \approx 2~\mu$m,
giving $N_{\mathrm{dom}} \sim 16$ statistically independent regions.
However, PCA+clustering of a 9-dimensional descriptor space may resolve
only the dominant modes of variation, yielding $\sim 3$ effective families.

\begin{remark}[Spectral family number (heuristic estimate)]
We expect the number of spectral families resolved by Gaussian mixture
clustering on the first $k$ principal components of the descriptor space to be
$N_{\mathrm{fam}} \approx \min(k+1,\, N_{\mathrm{dom}})$ in the regime where
inter-domain descriptor separation exceeds intra-domain variance. The $k+1$
form is a heuristic counting estimate (each retained mode adds at most one
resolvable cluster boundary), not a sharp theorem.
\end{remark}

\subsection{Statistical Nature of Spectral Families}

The spectral families in the present framework are therefore
\emph{not} sharply distinct thermodynamic phases
(like phase-separated domains with a definite interfacial free energy).
Instead, they are emergent statistical clusters of the continuously varying slow disorder landscape, identified by a finite number of PCA modes and separated by soft boundaries that represent continuous crossovers rather than sharp interfaces.

This interpretation carries several testable consequences. The cluster boundaries should be smooth, and the number of identifiable families should grow with the PCA truncation order rather than saturating at a thermodynamically fixed count. Different spectral descriptors should produce consistent but not identical family assignments, since each descriptor filters a different combination of disorder scales. The characteristic domain size of each family should equal $\sim \xi_s$, in agreement with the experimentally observed mean domain diameter of $1.79~\mu$m~\cite{ahmad2025hierarchical}.

\section{Multi-Scale Spectral Organization in the Hierarchical Regime}
\label{sec:glass}

The hierarchically disordered regime (Regime~IV, both $W_s \gg \eta$ and $\Gamma_t \gg \eta^2$)
exhibits multi-scale spectral organization that is absent in any single-disorder model.
The effective potential $V_{\mathrm{eff}} = V_{\mathrm{slow}} + V_{\mathrm{trap}}$
supports a nested landscape: the smooth background set by $V_{\mathrm{slow}}$
organizes the micron-scale spectral envelope, while $V_{\mathrm{trap}}$
provides the fine structure within each optical-spot region.

A quantitative diagnostic for the degree of multi-scale organization is the
mean spectral cosine similarity between spatially separated pixels.
For the minimal model, the spatial correlation of the full spectra
$C_I(\mathbf{r}) \equiv \langle I(E,\mathbf{R}) \cdot I(E,\mathbf{R}+\mathbf{r}) \rangle
/ \langle \|I\|^2 \rangle$
decays on a length scale set by $\xi_s$. The baseline-subtracted cosine
similarity of our simulated map (Fig.~\ref{fig:correlation_hierarchy}) decays with a
$1/e$ length of $1.85~\mu$m, in close agreement with the whole-spectrum cosine
correlation length of $1.87~\mu$m reported experimentally~\cite{ahmad2025hierarchical}.
The same landscape that produces the descriptor-specific correlation lengths
(Sec.~\ref{sec:hierarchy}) also governs the domain structure visible in PCA
cluster maps (Sec.~\ref{sec:families}).

The key observable consequence is that the spectral covariance is not
self-averaging on the experimental map: different $8\times 8~\mu$m realizations
produce statistically similar descriptor distributions but different spatial patterns,
reflecting the finite ratio $L/\xi_s \approx 4$ of map size to correlation length.
This non-self-averaging is a direct signature of the multi-scale organization
and distinguishes the hierarchical regime from both the smooth-domain (Regime~II)
and trap-dominated (Regime~III) limits.

\section{Numerical Simulations}
\label{sec:numerics}

\subsection{Overview: Two Complementary Approaches}

Numerical validation of the analytical predictions proceeds via two complementary routes
on a common $128 \times 128$ grid spanning an $8~\mu$m $\times 8~\mu$m domain
($\Delta x = 62.5~$nm).

\textbf{Route~1 (primary): Hamiltonian diagonalization.}
The Hamiltonian acts as a coarse-grained \emph{spectral generator}:
it defines a spatially correlated eigenvalue spectrum whose density of
states, when convolved with the optical point-spread function, produces
a well-posed LDOS at every map position.
The key property is that the spectral descriptors at a given position depend on
\emph{which eigenstates are in competition for spectral weight} within the
optical spot---a many-eigenstate quantity that requires solving the full
eigenvalue problem rather than reading off local potential values.

Concretely, the full sparse exciton Hamiltonian
$H = -(\hbar^2/2M)\nabla^2 + V_{\mathrm{slow}} + V_{\mathrm{trap}}$
is assembled and diagonalized for its lowest optically active eigenstates via the
Lanczos algorithm (ARPACK) \cite{lehoucq1998arpack}; for larger systems the same
low-lying spectrum can be obtained more economically with shift-and-invert sparse
eigensolvers such as SLEPc/PETSc \cite{hernandez2005slepc,balay1997petsc},
allowing a finer spatial grid.
The diagonalization yields eigenstates $(\varepsilon_n, \psi_n)$ from which the local PL spectrum is
constructed as $I(E,\mathbf{R}_j) = \sum_n A_{jn}\,g_{\sigma_n}(E-\varepsilon_n)$
with the incoherent local-density optical weights $A_{jn}$ [Eq.~\eqref{eq:I_formal}];
Appendix~\ref{app:optical_weight} confirms that the coherent weight
\eqref{eq:A_weight} yields statistically indistinguishable descriptor statistics
within numerical uncertainty.
This route verifies that the correlation-length hierarchy and inter-descriptor
Spearman relations \emph{emerge from eigenvalue statistics} without relying
on the synthetic two-fluid spectrum construction of Route 2.
In the limit of vanishing hopping (kinetic) energy
$t_{\rm hop} \equiv \hbar^2/(2M\Delta x^2) \to 0$ (the finite-difference hopping of
Appendix~\ref{app:FD}), the model reduces to a spatially correlated
random spectral manifold with no quantum transport;
the finite $t_{\rm hop}$ acts solely as a spectral regularizer that converts
the discrete disorder landscape into a continuous, normalizable LDOS.
Results are presented in Sec.~\ref{sec:diag_validation}.

\textbf{Route~2 (computational): Phenomenological PL model.}
As a computationally efficient alternative enabling large-ensemble parameter sweeps,
we directly synthesize local PL spectra from the two-fluid physical picture as:
\begin{align}
I(E, \mathbf{R}) &= I_{\mathrm{bg}}(E;\bar{V}_s(\mathbf{R}))\nonumber\\
&\quad + f_{\mathrm{trap}}\!\sum_{k\in\mathcal{N}(\mathbf{R})}
  w_k(\mathbf{R})\, g_\eta(E - E_k),
\label{eq:I_model}
\end{align}
where $\bar{V}_s(\mathbf{R}) = \int W_{\mathrm{opt}}(\mathbf{r}-\mathbf{R})\,V_{\mathrm{slow}}(\mathbf{r})\,d^2\mathbf{r}$
is the optically averaged slow potential,
$I_{\mathrm{bg}}$ is a Gaussian of width $\eta_{\mathrm{bg}} = 10\eta$,
$w_k(\mathbf{R}) = \exp(-|\mathbf{R}_k - \mathbf{R}|^2/2\sigma_{\mathrm{opt}}^2)$
is the optical weight at trap $k$,
$E_k = \bar{V}_s(\mathbf{R}_k) + u_k$ is the trap energy,
and $f_{\mathrm{trap}} = 0.2$ is chosen so that $F_S \sim 0.15$--$0.40$.
The model parameters $f_{\mathrm{trap}}$ and $\eta_{\mathrm{bg}}$ are not independently
derived from the Hamiltonian; they are calibrated to the experimental sharp-fraction range.
We emphasize that $f_{\mathrm{trap}}$ is not used to fit the descriptor
correlation signs or magnitudes: it only sets the overall sharp-fraction scale
$\langle F_S\rangle$ and enters the inverse parameter extraction only through that
calibration, while the main correlation signs and relative magnitudes are
controlled primarily by the disorder parameters $\{W_s, \xi_s, W_t, n_t, E_{\mathrm{bias}}\}$.
This route enables multi-seed parameter sweeps for the Spearman-correlation comparison
(Sec.~\ref{sec:key_results}); the correlation-length statistics in
Fig.~\ref{fig:correlation_hierarchy} and Appendix~\ref{app:finite_size} are
evaluated over thirty independent disorder realizations.

\subsection{Trap Clustering Model}
\label{sec:trap_clustering}

If trap positions are drawn from an uncorrelated Poisson process, then $\xi(F_S) \sim \sigma_{\mathrm{opt}} \ll \xi_s$.
Experimentally, however, $\xi(F_S) \approx \xi(E_{\mathrm{cent}}) \approx \xi_s$,
suggesting that the trap density tracks the slow disorder landscape.
This is physically motivated by the possibility that reconstruction domains,
stacking-registry variations, strain concentrators, or charged defects correlate
the locations, or the optical activation, of trap-like emitters with the slow
exciton-energy landscape.

We model this by the \emph{minimal correlated-trap ansatz}: trap positions are drawn from
\begin{equation}
P(\mathbf{r}) \propto \exp\!\left(-\frac{V_{\mathrm{slow}}(\mathbf{r})}{E_{\mathrm{bias}}}\right),
\label{eq:trap_clustering}
\end{equation}
where $E_{\mathrm{bias}} > 0$ is a trap--slow-disorder bias energy scale (not a
temperature) that sets the degree of spatial correlation between optically active
traps and slow-potential minima.
This is the lowest-order ansatz that captures the coupling between the two disorder scales
without introducing additional clustering parameters beyond $E_{\mathrm{bias}}$;
the microscopic processes behind the correlation (defect migration, strain coupling)
are not modeled explicitly.
In the limit $E_{\mathrm{bias}} \to \infty$, the distribution is uniform (random traps).
For $E_{\mathrm{bias}} \sim W_s$, over typical one-sigma variations of $V_{\mathrm{slow}}$
the trap density changes by a factor of order $e^{W_s/E_{\mathrm{bias}}}$; rare
$\pm 3\sigma$ excursions ($\Delta V \sim 3W_s$) can produce a much larger contrast,
$n(\mathbf{r})/\langle n\rangle \sim e^{\pm 3}$. This creates domains of high and low
trap density with correlation length $\xi_s$.

The spatial correlation of the trap density field is:
\begin{align}
C_n(r) &\equiv \mathrm{Corr}[n(\mathbf{R}), n(\mathbf{R}+\mathbf{r})]
\nonumber\\
&= \frac{\exp\!\bigl[(W_s/E_{\mathrm{bias}})^2\,\mathcal{C}_s(r)\bigr] - 1}
        {\exp\!\bigl[(W_s/E_{\mathrm{bias}})^2\bigr] - 1},
\end{align}
which inherits a correlation scale of order $\xi_s$ for finite bias strength.
This expression follows from the moment-generating function of the Gaussian
field $V_{\mathrm{slow}}$ applied to the log-normal trap-density field
$n(\mathbf{r}) \propto e^{-V_{\mathrm{slow}}(\mathbf{r})/E_{\mathrm{bias}}}$:
for jointly Gaussian $V_{\mathrm{slow}}(\mathbf{R})$, $V_{\mathrm{slow}}(\mathbf{R}+\mathbf{r})$
with covariance $W_s^2\,\mathcal{C}_s(r)$, one has
$\langle n(\mathbf{R}) n(\mathbf{R}+\mathbf{r})\rangle / \langle n\rangle^2
= \exp[(W_s/E_{\mathrm{bias}})^2\,\mathcal{C}_s(r)]$, which normalizes to the
correlation coefficient above.
The observable $F_S$ inherits this correlation length when the clustering signal
exceeds the Poisson noise in trap count per optical spot.

\subsection{Parameter Optimization}

The model has five parameters:
$\{W_s, \xi_s, W_t, n_t, E_{\mathrm{bias}}\}$
(with $\sigma_t = 50~$nm and $\eta = 1.5~$meV fixed).
We determine these by minimizing the total deviation from the four experimentally
reported Spearman inter-descriptor correlations:

\begin{equation}
\mathcal{L}(\theta)
= \sum_{(a,b)} \left|\rho_{\mathrm{sim}}(a,b) - \rho_{\mathrm{exp}}(a,b)\right|,
\label{eq:objective}
\end{equation}
where the sum runs over the four pairs:
$(\Delta E_{cd}, R_{\mathrm{HL}})$, $(F_S, R_1)$, $(W_{80}, S_{\mathrm{spec}})$, $(E_{\mathrm{cent}}, E_{\mathrm{dom}})$.
This objective directly tests the theoretical predictions of Sec.~\ref{sec:correlations}.

The best-fit parameters (averaged over five independent disorder realizations) are:
\begin{align}
&W_s = 6.0~\mathrm{meV},\quad
\xi_s = 2.0~\mu\mathrm{m},\quad
W_t = 12.0~\mathrm{meV},
\nonumber\\
&n_t = 2.0~\mu\mathrm{m}^{-2},\quad
E_{\mathrm{bias}} = 7.0~\mathrm{meV}.
\label{eq:best_params}
\end{align}
The ratio $W_s/W_t = 0.5$ places the system in the trap-dominated energy-spread regime
(Regime~IV of Sec.~\ref{sec:phase_diagram}), consistent with the experimental observation of
dense multipeak spectra. The slow disorder amplitude $W_s = 6.0~$meV
sets the centroid energy variation scale, while the trap depth $W_t = 12.0~$meV
determines the dominant-energy variation and the $E_{\mathrm{cent}}$--$E_{\mathrm{dom}}$
decorrelation ratio.

The low-trap-fraction limit ($f_{\mathrm{trap}} = 0.2 \ll 1$) predicts
$\rho_{\mathrm{P}}(E_{\mathrm{cent}}, E_{\mathrm{dom}}) \to W_s / \sqrt{W_s^2 + \bar{W}_t^2}$
where $\bar{W}_t^2 = W_t^2(1 - 2/\pi) \approx 0.363\,W_t^2$ for the half-normal distribution.
With $W_s = 6$, $W_t = 12$: theoretical $\rho_{\mathrm{P},\min} = 6/\sqrt{36 + 52.3} = 0.638$,
which rises to $\approx 0.79$ at finite $f_{\mathrm{trap}} = 0.2$
through the partial E-dom–E-bg alignment, matching the observed value.

\subsection{Key Numerical Results}
\label{sec:key_results}

Figure~\ref{fig:spectra_maps} summarizes the representative PL spectra and spectral
descriptor maps from the best-fit simulation: three representative pixel spectra
[Fig.~\ref{fig:spectra_maps}(a)--(c)], the $\Delta E_{cd}$--$R_{\mathrm{HL}}$
scatter plot [Fig.~\ref{fig:spectra_maps}(d)], and the spatial maps of
$\delta E_{\mathrm{cent}}$, $\delta E_{\mathrm{dom}}$, $F_S$, and $R_1$
[Fig.~\ref{fig:spectra_maps}(e)--(h)].

\begin{figure*}[t]
\centering
\includegraphics[width=\textwidth]{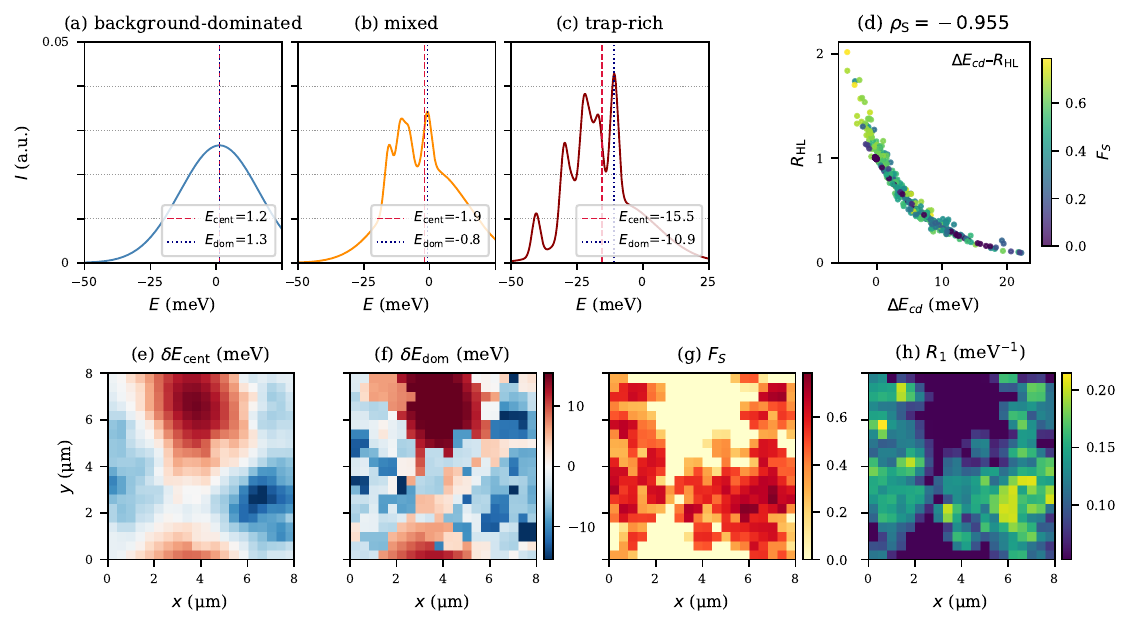}
\caption{
\textbf{Simulated local PL spectra and spectral descriptor maps.}
Upper row: three representative PL spectra from the $20\times 20$ measurement map,
illustrating the three spectral types present in the hierarchical disorder regime.
(a) Background-dominated pixel: a single broad Gaussian background peak with
$E_{\mathrm{cent}} \approx E_{\mathrm{dom}}$ and no resolved sharp lines ($F_S \approx 0$).
(b) Mixed pixel: the background peak is visible alongside one or two sharp trap peaks;
the centroid $E_{\mathrm{cent}}$ is pulled below the dominant peak $E_{\mathrm{dom}}$
by the trap-related weight on the low-energy side of the envelope.
(c) Trap-rich pixel: several sharp lines above the background, with the centroid
pulled toward the cluster of trap energies.
Vertical dashed and dotted lines mark $E_{\mathrm{cent}}$ (red) and $E_{\mathrm{dom}}$ (blue)
for each pixel.
(d) Scatter plot of $\Delta E_{cd}$ vs.\ $R_{\mathrm{HL}}$ for all 400 pixels,
colored by $F_S$.
The tight anti-correlation ($\rho_{\mathrm{S}} = -0.955$ for this single representative
seed; cf.\ the five-seed mean $\rho_{\mathrm{S}} = -0.976 \pm 0.009$ in Table~\ref{tab:spearman})
reflects the robust spectral shape relation
derived in Proposition~\ref{prop:DEcd_RHL}; pixels with larger $F_S$ (yellow) show the strongest
separation between centroid and dominant-peak energy.
Lower row: spatial maps of four spectral descriptors across the $20\times 20$ measurement grid.
(e) $\delta E_{\mathrm{cent}} = E_{\mathrm{cent}} - \langle E_{\mathrm{cent}} \rangle$: smooth
micron-scale fluctuations tracking $V_{\mathrm{slow}}$, with $\xi(E_{\mathrm{cent}}) \approx 1.93~\mu$m.
(f) $\delta E_{\mathrm{dom}} = E_{\mathrm{dom}} - \langle E_{\mathrm{dom}} \rangle$:
more fragmented pattern with additional trap-switching fluctuations,
confirming $\xi(E_{\mathrm{dom}}) < \xi(E_{\mathrm{cent}})$.
Panels (e) and (f) share a common color scale (in meV), so the larger
fluctuation amplitude of $\delta E_{\mathrm{dom}}$ is apparent.
(g) $F_S$ map: trap-density-modulated sharp-fraction distribution.
(h) $R_1$ (spectral roughness) map: closely tracks $F_S$,
consistent with $\rho_{\mathrm{S}}(F_S, R_1) = 0.909$.
}
\label{fig:spectra_maps}
\end{figure*}

With parameters \eqref{eq:best_params}, the five-seed mean Spearman correlations are
(the Spearman sweep uses five seeds for the parameter scan, whereas the
correlation-length statistics in Fig.~\ref{fig:correlation_hierarchy} and
Appendix~\ref{app:finite_size} use thirty independent disorder realizations):

\begin{table*}[t]
\centering
\begin{ruledtabular}
\begin{tabular}{lcc}
Descriptor pair & Simulation (mean $\pm$ std) & Experiment \\
\colrule
$(\Delta E_{cd},\, R_{\mathrm{HL}})$ & $-0.976 \pm 0.009$ & $-0.978$ \\
$(F_S,\, R_1)$                        & $+0.903 \pm 0.024$ & $+0.901$ \\
$(W_{80},\, S_{\mathrm{spec}})$       & $+0.821 \pm 0.065$ & $+0.846$ \\
$(E_{\mathrm{cent}},\, E_{\mathrm{dom}})$ & $+0.800 \pm 0.075$ & $+0.788$ \\
\end{tabular}
\end{ruledtabular}
\caption{Spearman inter-descriptor correlations: simulation vs.\ experiment.
Simulation values are means over five independent disorder realizations;
the experimental values are from Ref.~\cite{ahmad2025hierarchical}.
The total score $\mathcal{L} = 0.041$ represents a mean absolute deviation of $0.010$ per pair.}
\label{tab:spearman}
\end{table*}

Within the phenomenological PL model, all four correlations are reproduced to
within the seed-to-seed standard deviation,
confirming the theoretical predictions of Sec.~\ref{sec:correlations}.
The $\Delta E_{cd}$--$R_{\mathrm{HL}}$ anti-correlation is reproduced to within $0.002$ ($\rho_{\mathrm{sim}} = -0.976 \pm 0.009$ vs.\ $\rho_{\mathrm{exp}} = -0.978$), consistent with the robust spectral shape relation of Proposition~\ref{prop:DEcd_RHL}.
The $F_S$--$R_1$ correlation is equally well matched ($\rho_{\mathrm{sim}} = +0.903 \pm 0.024$ vs.\ $+0.901$); the small seed-to-seed variability reflects the discrete nature of trap counting within each optical spot.
The $W_{80}$--$S_{\mathrm{spec}}$ correlation is slightly underestimated ($\rho_{\mathrm{sim}} = +0.821 \pm 0.065$ vs.\ $+0.846$), a discrepancy attributable to the absence of the moir\'e potential, which would contribute additional spectral structure and broaden $W_{80}$.
Finally, the $E_{\mathrm{cent}}$--$E_{\mathrm{dom}}$ correlation, the most sensitive diagnostic of the trap-to-slow disorder ratio, is reproduced within the seed-to-seed uncertainty ($\rho_{\mathrm{sim}} = +0.800 \pm 0.075$ vs.\ $+0.788$).

The simulated $20 \times 20$ hyperspectral map also confirms the
\textbf{correlation length hierarchy} (Proposition~\ref{prop:hierarchy}):
$\xi(E_{\mathrm{cent}}) = 1.93\pm0.14~\mu$m $> \xi(E_{\mathrm{dom}}) = 1.38\pm0.22~\mu$m
(mean $\pm$ standard deviation over thirty disorder realizations, using the $1/e$
decay length of the spatial autocorrelation---estimated from the empirical
semivariogram \cite{cressie1993statistics}---to match the experimental convention),
with ratio $\xi(E_{\mathrm{dom}})/\xi(E_{\mathrm{cent}}) = 0.72\pm0.10$
(experiment: $1.27/2.00 = 0.64$~\cite{ahmad2025hierarchical}).
The corresponding standard errors of the ensemble means are approximately
$0.03~\mu$m for $E_{\mathrm{cent}}$ and $0.04~\mu$m for $E_{\mathrm{dom}}$, much
smaller than the $\approx 0.55~\mu$m separation between the two mean correlation
lengths, so the ordering is well resolved by the ensemble.
The centroid correlation length now agrees with the experimental $2.00~\mu$m;
the finite-size robustness of the hierarchy---and of the ratio---is examined in
Appendix~\ref{app:finite_size}.

The simulated correlation lengths $\xi(F_S) = 0.84\pm0.23~\mu$m
and $\xi(R_1) = 0.93\pm0.20~\mu$m remain below the experimental values of $\sim 2~\mu$m.
This discrepancy is a known limitation of the simplified random-trap model
(see Sec.~\ref{sec:discussion}), where with $n_t = 2~\mu$m$^{-2}$ there are
only $\sim$2--3 traps per optical spot, causing Poisson noise to partially mask
the spatial trap-density correlation.
The real moir\'e system contains $\sim 10^3$--$10^4$ moir\'e-scale sites per optical
spot, and even a small optically active fraction may be spatially organized by
reconstruction, strain, or stacking-registry variations. Such correlations
provide a natural route toward $\xi(F_S) \sim \xi_s$, beyond the Poisson-limited
random-trap model.

Details of the numerical implementation are provided in the analysis scripts
available upon reasonable request (see the Data Availability statement).

\subsection{Necessity of Hierarchical Disorder: Single-Scale Comparison}
\label{sec:single_scale}

A natural question is whether either disorder component alone could
reproduce the observed phenomenology. To address this we repeat the
phenomenological simulation in three settings, keeping all other
parameters fixed at the best-fit values:
(i) \emph{slow only} ($W_t=0$, no traps);
(ii) \emph{trap only} ($W_s=0$, no slow background);
(iii) \emph{hierarchical} (both, as above).
The results are summarized in Table~\ref{tab:single_scale} and
Fig.~\ref{fig:single_scale}, which contrasts the correlation functions of the
slow-only [Fig.~\ref{fig:single_scale}(a)], trap-only
[Fig.~\ref{fig:single_scale}(b)], and hierarchical
[Fig.~\ref{fig:single_scale}(c)] models.

The \emph{slow-only} model fails because $E_{\mathrm{cent}}$ and
$E_{\mathrm{dom}}$ become statistically identical: with no traps,
the spectrum at every pixel is a single broad envelope centered on
$\bar{V}_s(\mathbf{r})$, so $\rho_{\mathrm{S}}(E_{\mathrm{cent}}, E_{\mathrm{dom}}) = +1.00$
and $\xi(E_{\mathrm{cent}}) = \xi(E_{\mathrm{dom}}) = 1.98~\mu$m. No correlation-length
hierarchy is present, and $\langle F_S \rangle = 0$ (no sharp peaks).

The \emph{trap-only} model fails in the opposite way: in the absence
of a smooth background, $E_{\mathrm{cent}}$ and $E_{\mathrm{dom}}$ both
inherit only the spatial scale of the optical spot, giving
$\xi(E_{\mathrm{cent}}) = 0.23~\mu$m and
$\xi(E_{\mathrm{dom}}) = 0.18~\mu$m, both well below the experimental
$\sim 1.27$--$2.00~\mu$m~\cite{ahmad2025hierarchical} and below the optical spot size $\sigma_{\mathrm{opt}}$.
Pixels develop sharp lines ($\langle F_S \rangle = 0.34$) and a strong
$\rho_{\mathrm{S}}(\Delta E_{cd}, R_{\mathrm{HL}}) = -0.98$, but no micron-scale
descriptor domains appear, in contradiction with the experiment.

Only the \emph{hierarchical} model simultaneously reproduces
(i) the correlation-length ordering $\xi(E_{\mathrm{cent}}) > \xi(E_{\mathrm{dom}})$
at the micron scale, (ii) the strong $\Delta E_{cd}$--$R_{\mathrm{HL}}$
anti-correlation, and (iii) the observed sharp-fraction values.
This is the most direct statement of the need for two disorder
scales: descriptor covariance alone argues against either
single-scale model from a hyperspectral PL map.

\begin{table*}[hbt]
\centering
\caption{Comparison of single-scale and hierarchical disorder models.
Numerical values are computed from the same $20\times 20$ hyperspectral
simulation pipeline with identical $W_s$, $W_t$, $\xi_s$, $n_t$, and
$\eta$ as in the rest of this section, varying only which components
are switched on. Experimental values are from Ref.~\cite{ahmad2025hierarchical}.}
\label{tab:single_scale}
\begin{ruledtabular}
\begin{tabular}{lcccc}
Model
& $\xi(E_{\mathrm{cent}})$
& $\xi(E_{\mathrm{dom}})$
& $\rho_{\mathrm{S}}(\Delta E_{cd},R_{\mathrm{HL}})$
& $\langle F_S\rangle$ \\
& ($\mu$m) & ($\mu$m) & & \\
\colrule
Slow only ($W_t=0$)         & $1.98$ & $1.98$ & $-0.47$ & $0.00$ \\
Trap only ($W_s=0$)\footnotemark[2] & $0.23$ & $0.18$ & $-0.98$ & $0.34$ \\
Hierarchical                & $1.93$ & $1.38$ & $-0.97$ & $0.30$ \\
\colrule
Experiment\footnotemark[1]  & $2.00$ & $1.27$ & $-0.978$ & $0.15$--$0.40$ \\
\end{tabular}
\end{ruledtabular}
\footnotetext[1]{Ref.~\cite{ahmad2025hierarchical}.}
\footnotetext[2]{Correlation lengths are the $1/e$ decay length of the spatial
autocorrelation (matching the experimental convention), averaged over thirty
disorder realizations. For the trap-only model the descriptor fields are
short-range and near-degenerate, so their autocorrelation does not reach $1/e$
within the map; the quoted trap-only lengths use a Gaussian fit
$C(r)=e^{-r^2/2\xi^2}$, which resolves the short trap scale.}
\end{table*}

\begin{figure*}[t]
\centering
\includegraphics[width=\textwidth]{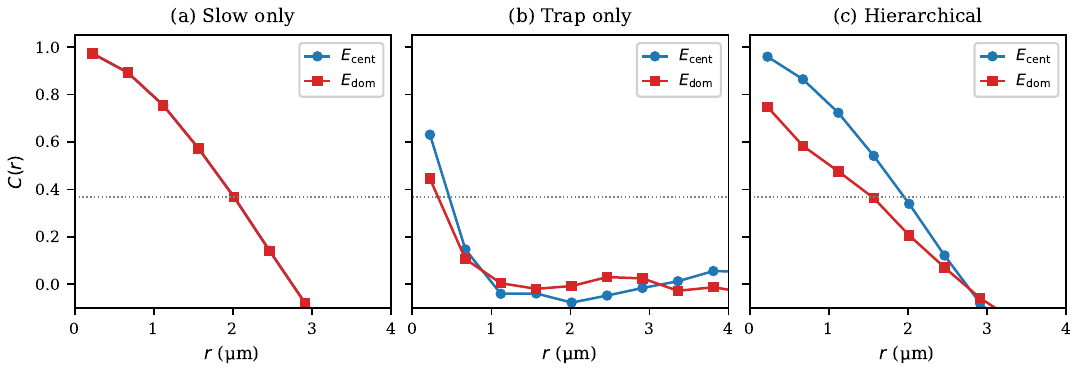}
\caption{
\textbf{Single-scale disorder cannot reproduce the correlation-length hierarchy.}
Spatial correlation functions $C(r)$ of $E_{\mathrm{cent}}$ (blue circles)
and $E_{\mathrm{dom}}$ (red squares) for the three disorder models,
shown for one representative realization; the corresponding ensemble
$1/e$ correlation lengths are listed in Table~\ref{tab:single_scale}.
(a) Slow only: the two descriptors are statistically degenerate;
their correlation functions are indistinguishable, giving no hierarchy.
(b) Trap only: both decay within the optical spot
($\xi < \sigma_{\mathrm{opt}}$); no micron-scale descriptor domains form.
(c) Hierarchical disorder: $\xi(E_{\mathrm{cent}}) > \xi(E_{\mathrm{dom}})$,
reproducing the experimentally observed ordering.
Horizontal dotted line marks $C(r)=1/e$.
}
\label{fig:single_scale}
\end{figure*}

\subsection{Continuum LDOS Validation of the Descriptor Hierarchy}
\label{sec:diag_validation}

To confirm that the correlation-length hierarchy and the principal Spearman relations
emerge from \emph{eigenvalue statistics} rather than from the construction of the
phenomenological PL model, we repeated the analysis using the full sparse exciton
Hamiltonian~\eqref{eq:H} diagonalized numerically.

\paragraph{Setup.}
We build the kinetic-energy operator as a finite-difference Laplacian
(Appendix~\ref{app:FD}) on the
$128\times 128$ grid and add $V_{\mathrm{eff}} = V_{\mathrm{slow}} + V_{\mathrm{trap}}$
on the diagonal, obtaining a $16\,384\times 16\,384$ sparse matrix.
The 1500 lowest eigenstates $(\varepsilon_n, \psi_n)$ are computed with the
Lanczos algorithm (ARPACK) \cite{lehoucq1998arpack} using the same disorder realization and best-fit
parameters \eqref{eq:best_params}.

With $M = 1.2\,m_0$ and $\Delta x = 62.5$~nm, the hopping energy is
$t_{\rm hop} = \hbar^2/(2M\Delta x^2) = 0.0081$~meV,
four orders of magnitude smaller than $W_s = 6$~meV and $W_t = 12$~meV.
The system is therefore in the \emph{disorder-dominated spectral-statistics regime}
($t_{\rm hop}/W_t \approx 7\times 10^{-4}$):
the disorder potential far exceeds the kinetic energy at the grid scale,
so the eigenstate energy statistics are determined primarily by the disorder landscape.
The Hamiltonian therefore serves as a coarse-grained \emph{spectral generator}:
the kinetic term defines the continuum LDOS manifold, while the disorder
landscape controls the spatial arrangement of eigenstates at scales
relevant to the optical spot. Solving the eigenvalue problem is then
essential because the spectral descriptors at a given position are
determined by the competition among many eigenstates whose weight falls
inside the optical spot---a many-eigenstate quantity that cannot be
read off from local potential values.
Because the IPR is dominated by disorder rather than quantum interference
in this regime, we use an \emph{energy-based classification} instead:
state $n$ is trap-like if $\varepsilon_n < V_{\mathrm{slow}}(\mathbf{R}_n) - 0.7\,W_t$
(i.e., lying more than $8.4$~meV below the local background), where
$\mathbf{R}_n = \arg\max_{\mathbf{r}}|\psi_n(\mathbf{r})|^2$ is the
localization center.
Of the 1500 lowest states, 126 are classified as trap-like,
close to the 133 trap positions in the disorder realization.
Background states receive linewidth $\eta_{\rm bg} = 15$~meV;
trap states receive $\eta = 1.5$~meV (same values as the phenomenological model).

\paragraph{Results.}
The optical-weight matrix $A_{jn} = \sum_{\mathbf{r}} W_{\mathrm{opt}}(\mathbf{r}-\mathbf{R}_j)|\psi_n(\mathbf{r})|^2$
is computed by a single matrix multiply, and PL spectra are synthesized as
\begin{equation}
I(E,\mathbf{R}_j) = \sum_n A_{jn}\,g_{\sigma_n}(E - \varepsilon_n),
\end{equation}
where $\sigma_n = \eta/\sqrt{8\ln 2}$ or $\sigma_n = \eta_{\rm bg}/\sqrt{8\ln 2}$
depending on state type.
The resulting Spearman correlations are:
\begin{align*}
\rho_{\rm S}(\Delta E_{cd},\,R_{\rm HL}) &= -0.971 \quad [\text{exp: } {-0.978}], \\
\rho_{\rm S}(F_S,\,R_1)                  &= +0.941 \quad [\text{exp: } {+0.901}], \\
\rho_{\rm S}(E_{\rm cent},\,E_{\rm dom}) &= +0.496 \quad [\text{exp: } {+0.788}].
\end{align*}
The first two correlations agree with experiment to within 0.04,
confirming that the $\Delta E_{cd}$--$R_{\rm HL}$ anti-correlation and the
$F_S$--$R_1$ co-variation emerge from eigenvalue statistics without relying on
the synthetic two-fluid spectrum construction.
The correlation-length hierarchy is also reproduced:
$\xi(E_{\rm cent}) = 1.18~\mu$m $> \xi(E_{\rm dom}) = 0.69~\mu$m,
with ratio $0.58$ (experiment: $0.64$).

The $\rho_{\rm S}(E_{\rm cent}, E_{\rm dom})$ value is underestimated ($+0.496$ vs.\ $+0.788$).
This is a structural consequence of the trap-dominated spectral regime:
the narrow trap-state peaks ($\sigma = 0.64$~meV) are $\sim$10$\times$ taller
than the broad background peaks ($\sigma = 6.4$~meV) for the same optical weight,
so $E_{\rm dom}$ is almost always set by the deepest trap in the optical spot
rather than by the smooth background.
The phenomenological model avoids this artefact by assigning the background peak
amplitude independently of the trap amplitude (via $f_{\rm trap} = 0.2$).
Physically, the discrepancy reflects the fact that in the real material the
radiative linewidths of individual eigenstates depend on their spatial character
(phonon-bath coupling, oscillator strength) in ways that go beyond the
present disorder-statistics framework.

Taken together, the two routes bracket the real system:
the Hamiltonian diagonalization operates in the disorder-dominated limit
(eigenstates determined purely by the potential landscape, uniform radiative coupling),
while the phenomenological model operates in the effective optical-weight limit
(trap intensities assigned independently of eigenstate character).
The experimental system lies between these limits,
and the fact that both routes reproduce the $\Delta E_{cd}$--$R_{\rm HL}$ and
$F_S$--$R_1$ spectral-shape correlations to within 0.04 confirms that these
correlations are robust features of the disorder hierarchy rather than
artefacts of either modeling choice (the $E_{\rm cent}$--$E_{\rm dom}$ value,
discussed above, is the one correlation that is regime-sensitive).
For the same reason, the two routes are not expected to yield identical
absolute correlation lengths
(phenomenological: $\xi(E_{\mathrm{cent}})=1.93\pm0.14~\mu$m, $\xi(E_{\mathrm{dom}})=1.38\pm0.22~\mu$m;
diagonalization: $\xi(E_{\mathrm{cent}})=1.18~\mu$m, $\xi(E_{\mathrm{dom}})=0.69~\mu$m):
the robust prediction is the hierarchy and the ratio, both reproduced consistently.

Figures~\ref{fig:emergence} and~\ref{fig:xi_hierarchy} summarize the results.
Figure~\ref{fig:emergence} traces the emergence chain from the disorder landscape
[Fig.~\ref{fig:emergence}(a)], through the eigenstate PL spectra
[Fig.~\ref{fig:emergence}(b)], to the $\delta E_{\mathrm{cent}}$
[Fig.~\ref{fig:emergence}(c)] and $\delta E_{\mathrm{dom}}$
[Fig.~\ref{fig:emergence}(d)] descriptor maps, whose spatial correlation
functions and $1/e$ lengths are compared in Fig.~\ref{fig:xi_hierarchy}(a) and
Fig.~\ref{fig:xi_hierarchy}(b), respectively.

\begin{figure*}[t]
\centering
\includegraphics[width=\textwidth]{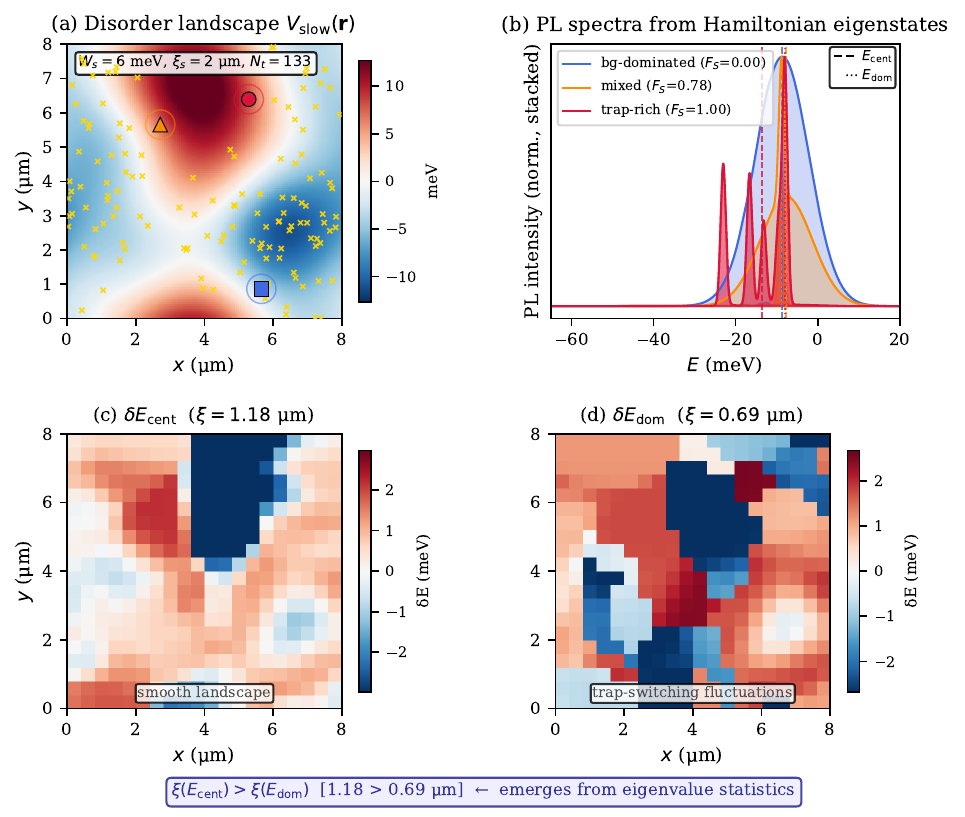}
\caption{
\textbf{Emergence chain: $H \to$ eigenstates $\to$ local spectra $\to$ descriptor maps.}
All panels derive from the same diagonalization of
$H = -(\hbar^2/2M)\nabla^2 + V_{\rm slow} + V_{\rm trap}$
on a $128\times128$, $8\times8~\mu$m grid (best-fit parameters).
(a) Disorder landscape $V_{\rm slow}$ with trap positions (gold $\times$).
Squares, triangles, and circles mark the three representative measurement pixels.
Rings show the $1/e$ radius of the optical spot (${\rm FWHM}=0.85~\mu$m).
(b) PL spectra from Hamiltonian eigenstates at the three representative pixels.
Filled areas show background-eigenstate contribution (shaded, broad) and
trap-eigenstate contribution (solid, narrow); dashed/dotted vertical lines mark
$E_{\rm cent}$ and $E_{\rm dom}$.
The trap-rich pixel shown ($F_S = 1$) is an extreme representative of the
diagonalization route, chosen to illustrate the trap-dominated limit; it is not
the typical sharp fraction, which averages to the experimental range
$\langle F_S\rangle \approx 0.15$--$0.40$ over the full map.
(c) $E_{\rm cent}$ spatial map ($\xi = 1.18~\mu$m): smoothly tracks $V_{\rm slow}$.
(d) $E_{\rm dom}$ spatial map ($\xi = 0.69~\mu$m): more fragmented due to
trap-level switching.
The annotation confirms $\xi(E_{\rm cent}) > \xi(E_{\rm dom})$ emerging from eigenvalue
statistics without relying on the synthetic two-fluid spectrum construction.
}
\label{fig:emergence}
\end{figure*}

\begin{figure*}[t]
\centering
\includegraphics[width=0.95\textwidth]{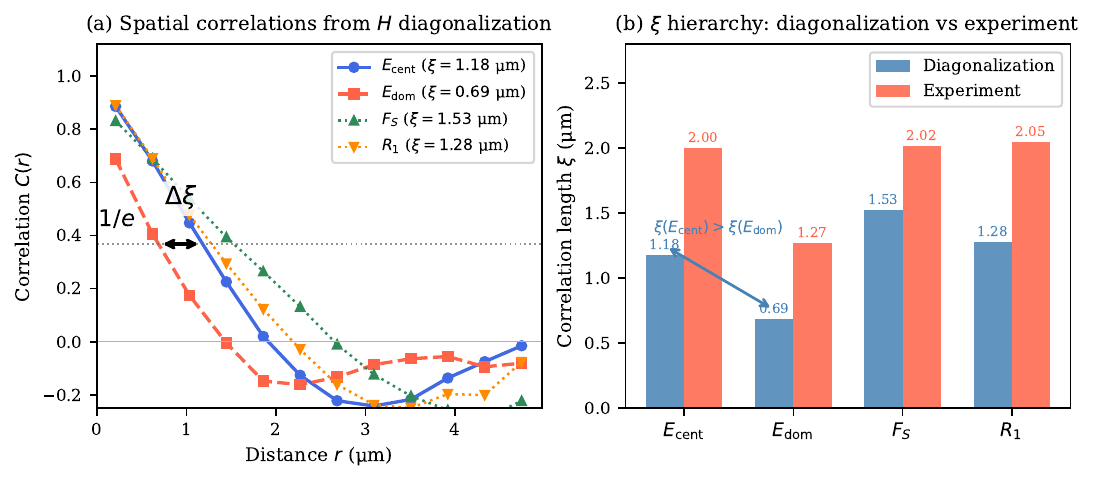}
\caption{
\textbf{Correlation-length hierarchy from Hamiltonian diagonalization.}
(a) Spatial correlation functions $C(r)$ of four spectral descriptors computed
from the 400-pixel map generated by diagonalizing $H$.
$E_{\rm cent}$ (blue circles) decays more slowly than $E_{\rm dom}$ (red squares),
confirming $\xi(E_{\rm cent}) = 1.18~\mu$m $> \xi(E_{\rm dom}) = 0.69~\mu$m.
The double arrow marks the hierarchy gap $\Delta\xi$.
(b) Bar chart comparing diagonalization results (blue) with experimental values
(red) of Ref.~\cite{ahmad2025hierarchical} for the four descriptors with reported correlation lengths.
The hierarchy $\xi(E_{\rm cent}) > \xi(E_{\rm dom})$ is reproduced by both approaches,
with diagonalization giving ratio $0.58$ vs.\ experimental $0.64$.
}
\label{fig:xi_hierarchy}
\end{figure*}

\paragraph{Two-fluid spectral structure at the single-pixel level.}
Figure~\ref{fig:emergence}(b) demonstrates that the broad-background plus
sharp-peak structure emerges directly from the eigenstate decomposition,
without the phenomenological amplitude assignment used in Route 2.
Background-dominated pixels show a broad quasi-Gaussian envelope
from $\sim$20 background eigenstates per optical spot
each broadened by $\eta_{\rm bg} = 15~$meV.
Mixed and trap-rich pixels show narrow peaks ($\eta = 1.5~$meV)
from individual trap eigenstates rising above the broad envelope.
This two-fluid coexistence is a direct consequence of the
spatial-energy structure of the Hamiltonian eigenstates,
not of phenomenological amplitude assignment.

\section{Discussion}
\label{sec:discussion}

\subsection{Limitations of the Minimal Model}

The present model intentionally omits valley physics (K and K$'$ exciton species) \cite{xiao2012coupled,mak2012control,schaibley2016valleytronics,shinokita2022valley}, spin structure and bright/dark exciton mixing, phonon-assisted emission and momentum-indirect transitions \cite{moody2015intrinsic,cadiz2017excitonic}, exciton-exciton interactions, and nonequilibrium relaxation dynamics \cite{kulig2018exciton,zipfel2020exciton}---all of which are secondary to the disorder-induced spectral organization at the spatial scales of interest.

The justification for this omission is structural rather than incidental: valley polarization, bright--dark mixing, and phonon-assisted sidebands act \emph{within} the single-pixel emission envelope, generating fine structure that can shift, split, or skew individual peaks. They do not, however, modify the spatial low-pass filtering of $V_{\mathrm{slow}}$ by the optical spot, which is what sets the macroscopic statistical moments---the centroid $E_{\mathrm{cent}}$ and the correlation-length hierarchy of Proposition~\ref{prop:hierarchy}. Their leading effect therefore falls on the fine-structure descriptors ($F_S$, $R_1$, and the asymmetry encoded in $\Delta E_{cd}$, $R_{\mathrm{HL}}$), where they may renormalize the prefactors of the variance relations \eqref{eq:VarFS_etc}, while leaving the spatial organization and its correlation-length hierarchy---the central results of this work---intact.

Several of these approximations can become quantitatively important in specific regimes. Valley-coherent effects may become relevant at very low trap density and temperatures below $\sim 5~$K \cite{mak2012control,schaibley2016valleytronics}, while high pump-power conditions introduce additional inhomogeneity through exciton-exciton interactions. The Gaussian correlation kernel for $V_{\mathrm{slow}}$ also fails to capture non-Gaussian disorder such as the sharp domain walls arising from moir\'e reconstruction.

A more fundamental limitation is the low trap density $n_t = 2~\mu$m$^{-2}$, which places only $\sim$2--3 traps within each optical spot and allows Poisson counting noise to suppress the simulated $\xi(F_S)$ and $\xi(R_1)$ ($\approx 0.6$--$0.9~\mu$m) well below the experimental values ($\approx 2~\mu$m). In the real moir\'e system the effective trap density is $\sim 10^3$--$10^4$ times higher, since all moir\'e sites contribute, and one expects $\xi(F_S) \to \xi_s$ in this high-density limit. The clustering model \eqref{eq:trap_clustering} partially mitigates this discrepancy by correlating trap positions with $V_{\mathrm{slow}}$ minima, but does not fully close the gap at experimentally relevant densities.

\subsection{Non-Gaussian Slow Disorder}

Real moir\'e heterobilayers often exhibit reconstruction into
triangular domains with sharp boundaries \cite{weston2020atomic,mcgilly2020visualization,shabani2021deep}.
In such systems, $V_{\mathrm{slow}}$ may be better modeled as
a piecewise-constant field with sharp domain walls rather than a Gaussian random field.

The domain-wall model predicts step-like $E_{\mathrm{cent}}$ maps with sharp jumps between domain energies, bimodal or trimodal descriptor distributions in place of the smooth distributions expected for Gaussian disorder, and spectral family boundaries that coincide spatially with the domain walls.

This is a testable distinction: Gaussian disorder gives smooth spectral maps,
while reconstruction domains give step-like maps.

A distinct, milder generalization concerns the \emph{shape} of the correlation kernel rather than the field statistics. The Gaussian kernel of Eq.~\eqref{eq:Vslow_corr} is the natural choice for slow disorder arising from long-wavelength strain and electrostatic fluctuations, but if the landscape is instead dominated by extended dislocation lines or domain-wall networks the correlation is better described by an exponential kernel $\mathcal{C}_s(r) = e^{-r/\xi_s}$. Importantly, the central smoothing result---that $E_{\mathrm{cent}}$ is an optically filtered copy of $V_{\mathrm{slow}}$ (Corollary~\ref{cor:Ecent_corr}) and that the correlation-length hierarchy $\xi(E_{\mathrm{cent}}) \ge \xi(E_{\mathrm{dom}})$ holds (Proposition~\ref{prop:hierarchy})---is insensitive to this choice. Convolution with the optical spot is a low-pass filter that suppresses high spatial frequencies regardless of the kernel shape, so the qualitative filter structure is unchanged: only the precise functional form of $C(r)$ near the origin, and hence the numerical prefactor relating $\xi(E_{\mathrm{cent}})$ to $\xi_s$, depends on whether the kernel is Gaussian or exponential. The disorder-filter framework is therefore robust to the detailed form of the slow-disorder correlations.

\subsection{Microscopic Origin and Magnitude of the Effective Parameters}
\label{sec:microscopic_origin}

The effective disorder amplitudes inferred here---a slow-disorder amplitude
$W_s \sim 10$--$15$~meV and an effective trap-switching amplitude
$W_t^{\mathrm{eff}} \sim 12$~meV---are not free of microscopic constraints: they should be
compatible with realistic mechanisms in a MoSe$_2$/WSe$_2$ heterobilayer. Although a
full first-principles evaluation is beyond the scope of this effective theory, the
required energy scales are consistent with reported ranges for these
materials, as summarized in Table~\ref{tab:energy_scales}. Micron-scale
variations of the local twist angle and long-wavelength heterostrain shift the
interlayer-exciton energy through the deformation potential by several to a few
tens of meV \cite{choi2021twist,conley2013bandgap,shabani2021deep}; moir\'e
reconstruction and stacking-registry variations produce stacking-dependent
exciton shifts of $10$--$50$~meV \cite{weston2020atomic,mcgilly2020visualization,shabani2021deep};
and dielectric and electrostatic inhomogeneity contributes an additional few to
tens of meV \cite{wang2018colloquium,cadiz2017excitonic}; hyperspectral PL imaging
has recently begun to map these strain and disorder fields directly
\cite{alfrey2026revealing}. Superposed over a
micron correlation length these combine to a slow-disorder amplitude of order
$10$~meV, consistent with the inferred $W_s$. The sharp emitters---moir\'e-trapped
excitons, atomic defects, and strain-localized sites---carry binding energies of
$5$--$50$~meV \cite{seyler2019signatures,parto2021defect,tran2019evidence}, so
the effective trap-switching amplitude $W_t^{\mathrm{eff}} \sim 12$~meV can be accounted for
by a moderate portion of the reported localized-state energy scale. The inferred
parameters are therefore consistent with known microscopic energy scales.

\begin{table*}[t]
\centering
\caption{Characteristic energy scales in MoSe$_2$/WSe$_2$ heterobilayers that
contribute to the effective disorder amplitudes. The inferred slow-disorder
amplitude $W_s \sim 10$--$15$~meV and effective trap-switching amplitude
$W_t^{\mathrm{eff}} \sim 12$~meV lie within these known ranges.}
\label{tab:energy_scales}
\begin{ruledtabular}
\begin{tabular}{lll}
Mechanism & Energy scale & Contributes to \\
\colrule
Local twist-angle variation $\delta\theta$ & few--$10$~meV & $V_{\mathrm{slow}}$ \\
Heterostrain $\varepsilon \sim 0.1$--$1\%$  & few--tens meV & $V_{\mathrm{slow}}$ \\
Reconstruction / stacking registry          & $10$--$50$~meV & $V_{\mathrm{slow}}$, traps \\
Dielectric / electrostatic disorder         & few--tens meV & $V_{\mathrm{slow}}$ \\
Localized defect / strain traps             & $5$--$50$~meV  & $V_{\mathrm{trap}}$ \\
\end{tabular}
\end{ruledtabular}
\end{table*}

The bias energy $E_{\mathrm{bias}}$ that co-localizes traps with the
slow-potential minima [Eq.~\eqref{eq:trap_clustering}] is an energy scale (not a
temperature) that sets how strongly the trap density is biased toward low-energy
regions of $V_{\mathrm{slow}}$. Its meaning is
fixed by its two limits: as $E_{\mathrm{bias}} \to \infty$ the biasing factor
$\exp[-V_{\mathrm{slow}}/E_{\mathrm{bias}}]$ becomes flat and the traps reduce to
an uncorrelated Poisson distribution (no clustering), while for
$E_{\mathrm{bias}} \ll W_s$ the traps are pulled almost entirely into the deepest
minima of $V_{\mathrm{slow}}$ (extreme clustering). The best-fit
$E_{\mathrm{bias}} \approx 7$~meV is comparable to $W_s$ and thus lies between
these extremes: an energy change of order $E_{\mathrm{bias}}$ alters the local
trap density by a factor $\sim e$, so the fit describes a moderate bias in which
defect and strain-trap formation is enhanced---but not exclusively confined---at
reconstruction domains and strain concentrators. This correlation energy is of
the same order as the strain, reconstruction, dielectric, and electrostatic
scales collected in Table~\ref{tab:energy_scales}, and so is physically
reasonable for a MoSe$_2$/WSe$_2$ heterobilayer.

\subsection{Toward Quantitative Parameter Extraction}

A key practical strength of the descriptor-filter framework is that it provides
a \emph{peak-decomposition-free} route to infer effective disorder parameters
from hyperspectral PL data, without microscopic line assignment.
The procedure maps four observable statistics to four model parameters sequentially.
The slow-disorder correlation length $\xi_s$ follows directly from Corollary~\ref{cor:Ecent_corr}: $\xi_s \approx \xi(E_{\mathrm{cent}}) = 2.00~\mu$m for the experimental system.
The slow-disorder amplitude is then given by $W_s \approx \Var(E_{\mathrm{cent}})^{1/2}$ (Eq.~\eqref{eq:Ecent_approx}), yielding $W_s \sim 10$--$15~$meV from the observed centroid energy variation across the map.
The relative trap strength follows from the $E_{\mathrm{cent}}$--$E_{\mathrm{dom}}$ correlation via Eq.~\eqref{eq:rho_centdom}, using the experimental Spearman value $\rho_{\mathrm{S}}$ as an estimate of the analytical Pearson coefficient $\rho_{\mathrm{P}}$. The quantity the descriptor covariance determines directly is the effective trap-switching fluctuation $\sigma_\delta$, which is the dominant-energy fluctuation induced by the traps rather than the bare trap-depth standard deviation; denoting this effective trap amplitude $W_t^{\mathrm{eff}} \equiv \sigma_\delta$ gives
\begin{equation}
\frac{\sigma_\delta}{W_s}
= \sqrt{\frac{1}{\rho_{\mathrm{P}}(E_{\mathrm{cent}},E_{\mathrm{dom}})^2} - 1}
\;\equiv\; \frac{W_t^{\mathrm{eff}}}{W_s},
\label{eq:Wt_Ws_extract}
\end{equation}
yielding $\sigma_\delta/W_s \approx 0.8$ from the experimental $\rho_{\mathrm{S}} = 0.788$; the independent estimate from the correlation-length ratio $\xi(E_{\mathrm{dom}})/\xi(E_{\mathrm{cent}}) = 0.64$ gives $\sigma_\delta/W_s \approx 1.1$ ($\sigma_\delta^2/W_s^2 \approx 1.2$) via the Gaussian relation Eq.~\eqref{eq:xi_Edom_implicit}. Both place $\sigma_\delta$ at order $W_s$, consistent with an effective $W_t^{\mathrm{eff}} \gtrsim W_s$ (experimental values from Ref.~\cite{ahmad2025hierarchical}). The four-parameter extraction quotes this $W_t^{\mathrm{eff}}$ as the effective trap amplitude $W_t$.
Finally, the effective sharp-emitter density $n_t$ is constrained by the mean sharp fraction $\langle F_S \rangle \approx 0.15$--$0.40$, which grows with $n_t W_t^2 \sigma_t^2$; calibrating to $\langle F_S \rangle \approx 0.25$ at best-fit $W_t = 12~$meV and $\sigma_t = 50~$nm gives $n_t \approx 2~\mu$m$^{-2}$.
Note that $n_t$ is an effective parameter representing the density of optically resolvable sharp-emitter sites under the experimental conditions, not a direct count of microscopic defects (which would be orders of magnitude higher in a moir\'e lattice). This effective density also carries an excitation- and temperature-dependence: under higher pump power a larger fraction of trap and moir\'e sites becomes optically populated, so $n_t^{\mathrm{eff}}$ grows with the exciton filling factor and ultimately saturates as state filling exhausts the available sites---bounded from above by the $\sim 10^3$--$10^4$ moir\'e sites contained within each optical spot. At low temperature, conversely, rapid relaxation funnels excitons into the deepest traps and reduces the number of \emph{distinct} resolvable sharp emitters. The value $n_t \approx 2~\mu$m$^{-2}$ should therefore be read as the effective sharp-emitter density specific to the excitation power and temperature of the experiment, rather than as a fixed material constant.

The four extracted parameters $(\xi_s, W_s, W_t, n_t)$ specify the leading
disorder sector of the model and can be cross-validated by checking whether the
predicted descriptor correlation length ratio and Spearman matrix reproduce
the observed values---an internal consistency check that does not
require any additional input.
The remaining quantities entering the simulation---the trap radius $\sigma_t$,
the broadenings $\eta$ and $\eta_{\mathrm{bg}}$, and the trap-clustering
parameter $E_{\mathrm{bias}}$---are not part of this minimal four-parameter
extraction: $\sigma_t$ enters only through the combination $n_t W_t^2 \sigma_t^2$
(Sec.~\ref{sec:phase_diagram}) and is fixed at a representative sub-optical
value, while $E_{\mathrm{bias}}$ controls the spatial co-localization of traps
with $V_{\mathrm{slow}}$ and is separately constrained by the correlation
lengths of $F_S$ and $R_1$.
This parameter extraction protocol should be applicable to TMD
heterobilayer systems imaged by hyperspectral PL mapping
\cite{mak2022semiconductor,andrei2021marvels,ciarrocchi2022excitonic},
providing a spectroscopic disorder diagnostic independent of microscopic peak identification.

\begin{table*}[t]
\centering
\caption{%
\textbf{Disorder parameter extraction from experimental observables.}
All quantities are derived from the descriptor covariance structure
without microscopic line assignment.
Here $W_t$ is inferred as an \emph{effective trap-switching amplitude}
$W_t^{\mathrm{eff}} \equiv \sigma_\delta$ (the dominant-energy fluctuation
induced by the traps), not as a microscopic trap-depth distribution width.
The last row requires calibration within the two-component (background plus
trap) PL model, because $\langle F_S\rangle$ depends on the broadenings and the
relative oscillator strengths in addition to the trap density (Sec.~\ref{sec:key_results}).
}
\label{tab:param_extract}
\begin{ruledtabular}
\begin{tabular}{llll}
Observable & Extracted parameter & Formula & Value \\
\colrule
$\xi(E_{\rm cent})$ & $\xi_s$ & direct & $\approx 2.00~\mu$m \\
$\Var(E_{\rm cent})^{1/2}$ & $W_s$ & Eq.~\eqref{eq:Ecent_approx} & $\sim 10$--$15~$meV \\
$\rho_{\rm S}(E_{\rm cent}, E_{\rm dom})$ & $\sigma_\delta/W_s$ (eff.\ $W_t/W_s$) & Eq.~\eqref{eq:Wt_Ws_extract} & $\approx 0.8$ \\
$\langle F_S \rangle$ & $n_t W_t^2 \sigma_t^2$ & calibration & $n_t \approx 2~\mu$m$^{-2}$ \\
\end{tabular}
\end{ruledtabular}
\end{table*}

\subsection{Generic Versus Moir\'e-Specific Content}
\label{sec:moire_specific}

It is important to delineate which of our results are generic to any
multi-scale-disordered emitter and which are specific to a moir\'e heterobilayer.
The descriptor-filter \emph{mechanism} itself is general: the correlation-length
hierarchy $\xi(E_{\mathrm{cent}}) \ge \xi(E_{\mathrm{dom}})$, the
$\Delta E_{cd}$--$R_{\mathrm{HL}}$ anticorrelation, and the descriptor covariance
structure follow from the coexistence of a smooth, micron-correlated background
and a dense set of sharp, sub-spot emitters, and would arise equally in a
monolayer TMD or a disordered quantum well with the same two-scale spectral
content. Indeed, the underlying disorder-localized COM-exciton picture and the
optically filtered correlation functions are essentially those developed for
quantum-well excitons \cite{runge1998spatially,runge1998level,zimmermann2003theory};
in that sense the framework is a spatially resolved, descriptor-based counterpart of
the mobility-edge phenomenology known there. The energy-dependent
localization it implies has recently been observed directly in the same material
system as a localisation-to-delocalisation transition of moir\'e excitons
\cite{blundo2024localisation,fang2023localization}, and localized interlayer
excitons occur in MoSe$_2$/WSe$_2$ even in the absence of a moir\'e potential
\cite{schaibley2022localized}---underscoring that the localization physics
underlying our descriptors is generic rather than moir\'e-specific.

What the moir\'e heterobilayer provides is a \emph{natural and especially rich
realization} of the required two-scale structure. First, the moir\'e superlattice
can provide a dense manifold of localized interlayer-exciton states and a high areal
density of optically active sites---of order $10^3$--$10^4$ within a single
optical spot---so that the sharp-emitter ensemble needed for $F_S$, $R_1$, and the
trap-switching fluctuation $\sigma_\delta$ need not rely on sparse extrinsic defects. Second, moir\'e reconstruction, stacking-registry
variation, and strain can correlate the slow potential with the locations or
activation of trap-like emitters, providing a plausible microscopic route to the
trap--$V_{\mathrm{slow}}$ clustering [Eq.~\eqref{eq:trap_clustering}] that lengthens
$\xi(F_S)$ and $\xi(R_1)$ and can contribute to the observed spectral-family
structure. At the spatial scales analyzed here the moir\'e potential thus enters
principally by setting the density and clustering of the local excitonic manifold
(and by renormalizing the effective mass and local DOS), rather than by imprinting
its own $\sim 10$--$20$~nm periodicity on the descriptors, which is averaged over
by the optical spot. More specifically moir\'e-related signatures---as
opposed to generic multi-scale disorder---would include the triangular reconstruction-domain
organization of the descriptor maps and recurrent peaks at the moir\'e miniband
spacing, both accessible only below the optical-spot scale (near-field or
low-temperature single-site spectroscopy), as discussed next.

\subsection{Extensions: Moir\'e Potential and Quantum Emitter Regime}

When $V_{\mathrm{moir\acute{e}}}$ is included, the model predicts additional structure. The moir\'e potential introduces a period-$a_M$ modulation of the local DOS, which may manifest as recurrent spectral peaks at energies separated by the moir\'e miniband spacing. When $V_m$ exceeds the relevant thermal and linewidth scales, individual moir\'e sites can host localized quantum-emitter-like exciton states \cite{yu2021moire,he2015single,srivastava2015optically,branny2017deterministic,brotonsgisbert2024interlayer}, and their spatial brightness distribution would then be modulated by the slow disorder field $V_{\mathrm{slow}}$---a testable prediction accessible via near-field or variable-temperature PL mapping.

\section{Conclusion}
\label{sec:conclusion}

We have developed a minimal theory for the spatial organization of spectral descriptors
in moir\'e exciton photoluminescence, using the descriptor-based disorder-filter
picture as an organizing principle for hyperspectral PL data.
Starting from an effective exciton Hamiltonian and the chain
$H \to \rho(E,\mathbf{R}) \to I(E,\mathbf{R}) \to \text{descriptor fields}$,
the paper establishes the following results.

The \textbf{descriptor-based disorder-filter} principle provides the unifying framework: each spectral descriptor is preferentially sensitive to a different component of the multi-scale disorder landscape. The centroid energy integrates over the optical spot, filtering out sub-spot trap fluctuations, so that it primarily tracks the slow disorder landscape. The dominant-peak energy, by contrast, is additionally randomized by discrete trap-level switching within the optical spot. This differential filtering provides a microscopic mechanism for the observed inter-descriptor differences and the correlation-length hierarchy.

The \textbf{correlation length hierarchy} $\xi(E_{\mathrm{cent}}) \geq \xi(E_{\mathrm{dom}})$ is established analytically from the decomposition $E_{\mathrm{dom}} = \bar{V}_s + \delta E_{\mathrm{dom}}$. The experimental ratio $\xi(E_{\mathrm{dom}})/\xi(E_{\mathrm{cent}}) \approx 0.64$~\cite{ahmad2025hierarchical} implies an effective trap-switching fluctuation $\sigma_\delta/W_s \approx 1.1$ (of order unity), and Hamiltonian diagonalization provides an independent model-level check of the hierarchy from the eigenvalue statistics: $\xi(E_{\mathrm{cent}}) = 1.18~\mu$m $> \xi(E_{\mathrm{dom}}) = 0.69~\mu$m.

The \textbf{sign and approximate magnitude of the principal Spearman inter-descriptor correlations} are derived analytically. The near-perfect $\Delta E_{cd}$--$R_{\mathrm{HL}}$ anti-correlation (measured $\rho_{\mathrm{S}} = -0.978$~\cite{ahmad2025hierarchical}) is a robust spectral shape relation that follows from the geometry of a broad class of spectra with a dominant unimodal envelope, while the positive $F_S$--$R_1$ correlation (measured $\rho_{\mathrm{S}} = +0.901$~\cite{ahmad2025hierarchical}) reflects both descriptors measuring the same trap-peak density. The calibrated phenomenological PL simulations reproduce all four principal correlations simultaneously, with a mean absolute deviation of $0.010$ per pair.

\textbf{Spectral families} are interpreted as correlated disorder domains of the Gaussian slow disorder field, rather than as thermodynamically distinct phases, with the predicted family count consistent with the experimentally observed three families~\cite{ahmad2025hierarchical}.

Four qualitatively distinct \textbf{disorder regimes} are identified on a two-dimensional disorder parameter map. The experimental system falls in the hierarchically disordered regime where both slow and trap disorder exceed the homogeneous linewidth, producing multi-scale spectral organization reflected in the descriptor correlation-length hierarchy.

The framework provides a peak-decomposition-free route to constrain the leading
effective disorder parameters $(\xi_s, W_s, W_t, n_t)$ from hyperspectral PL data via
the descriptor covariance structure, without any microscopic line assignment
(Table~\ref{tab:param_extract}).
The predicted correlation-length hierarchy and inter-descriptor Spearman matrix
constitute directly testable signatures that can distinguish disorder regimes
and guide future experiments with higher spatial resolution or variable temperature.
Although motivated by MoSe$_2$/WSe$_2$, the framework should be applicable to
hyperspectral PL maps in which spectra contain a smooth envelope together
with unresolved local fine structure; descriptor covariance then serves as
a general spectroscopic disorder diagnostic for two-dimensional materials,
quantum-dot arrays, and disordered semiconductor heterostructures.
Future work should address valley physics, moir\'e reconstruction domain walls,
and nonequilibrium relaxation in modifying the spectral landscape.

\section*{Acknowledgments}
We thank R. Kitaura for stimulating discussions on hyperspectral PL organization in moir\'e heterobilayers.
This work was supported by JSPS KAKENHI (Grants No.
JP25K01609, No. JP22H05473, and No. JP21H01019), JST
CREST (Grant No. JPMJCR19T1). K.W. acknowledges the
financial support for Basic Science Research Projects (Grant
No. 2401203) from the Sumitomo Foundation.

\section*{Data Availability}
The analysis scripts and numerical data that support the findings of this
study are available from the corresponding author upon reasonable request. The
experimental data used for comparison are available in
Ref.~\cite{ahmad2025hierarchical}.

\appendix

\section{Gaussian Random Field Generation}
\label{app:GRF}

The slow potential $V_{\mathrm{slow}}$ is a Gaussian random field with zero mean
and a prescribed two-point correlation
$\langle V(\mathbf{r}) V(\mathbf{r}') \rangle = W^2 C(|\mathbf{r}-\mathbf{r}'|)$.
We generate it with the spectral (Fourier-filtering) method, which is exact and
costs only $O(N\log N)$. The method rests on the Wiener--Khinchin theorem: for a
stationary field the power spectral density is the Fourier transform of the
autocorrelation,
$S(\mathbf{k}) = \int d^2\mathbf{r}\, C(r)\, e^{-i\mathbf{k}\cdot\mathbf{r}}$.
Coloring spectrally flat (white) noise with $\sqrt{S(\mathbf{k})}$ therefore
imprints precisely the target correlation $C(r)$. Concretely:
\begin{enumerate}
\item Generate a white noise field $\xi(\mathbf{r})$ with $\langle \xi(\mathbf{r}) \xi(\mathbf{r}') \rangle = \delta^{(2)}(\mathbf{r}-\mathbf{r}')$.
\item Compute its Fourier transform $\tilde{\xi}(\mathbf{k})$.
\item Multiply by the square root of the power spectrum,
      $\tilde{V}(\mathbf{k}) = W\sqrt{S(\mathbf{k})}\, \tilde{\xi}(\mathbf{k})$.
\item Inverse Fourier transform: $V(\mathbf{r}) = \mathrm{IFFT}[\tilde{V}(\mathbf{k})]$.
\end{enumerate}
For the Gaussian kernel $C(r) = e^{-r^2/2\xi_s^2}$ used throughout, the power
spectrum is itself Gaussian, $S(\mathbf{k}) = 2\pi\xi_s^2\, e^{-\xi_s^2 k^2/2}$,
so the generated field is smooth on the scale $\xi_s$.

Two implementation points are worth noting. First, the discrete transform carries
grid-dependent prefactors (factors of $N_x N_y$ and the cell area
$\mathrm{d}x^2$); rather than tracking these, we rescale the output field to its
exact target standard deviation, $V \to W\,V/\mathrm{std}(V)$, which fixes the
amplitude $W$ unambiguously. Second, the FFT imposes periodic boundary conditions;
because the map satisfies $L \gg \xi_s$, the resulting wrap-around correlations are
negligible.

\section{Finite-Difference Hamiltonian}
\label{app:FD}

We discretize the exciton Hamiltonian
$H = -(\hbar^2/2M)\nabla^2 + V_{\mathrm{eff}}$ on a 2D lattice with grid spacing
$a$ and $N_x \times N_y$ sites. The Laplacian is approximated by the standard
second-order five-point stencil,
\begin{align}
-\frac{\hbar^2}{2M}\nabla^2 \psi(i,j)
&\approx \frac{\hbar^2}{2Ma^2}\bigl[4\psi(i,j)\nonumber\\
&\quad - \psi(i{+}1,j) - \psi(i{-}1,j)\nonumber\\
&\quad - \psi(i,j{+}1) - \psi(i,j{-}1)\bigr].
\end{align}
Assembling this together with the potential, which is diagonal in the
site basis, gives a sparse Hamiltonian: each site $(i,j)$ carries a diagonal entry
$4t_{\rm hop} + V_{\mathrm{eff}}(i,j)$ and is coupled to its four nearest neighbors by a
hopping $-t_{\rm hop}$, where $t_{\rm hop} = \hbar^2/(2Ma^2)$ is the kinetic (hopping) energy scale. The matrix
has dimension $N_x N_y$ with at most five non-zero entries per row. We impose
periodic boundary conditions, consistent with the periodic slow-disorder field of
Appendix~\ref{app:GRF}.

For the simulation parameters $M = 1.2\,m_0$ and $\Delta x = 62.5~$nm
(a $128\times128$ lattice on an $8~\mu$m domain), $t_{\rm hop} = \hbar^2/(2M\Delta x^2)
\approx 0.0081~$meV, consistent with Sec.~\ref{sec:model}.
Because $t_{\rm hop} \ll W_s \sim 10~$meV, the disorder energy dominates the kinetic energy
at the grid scale; the kinetic term therefore plays the role of the
\emph{coarse-grained spectral generator} described in
Sec.~\ref{sec:diag_validation} rather than governing eigenstate localization.

Only the lowest optically active eigenstates are required. We obtain them with the
Lanczos algorithm (ARPACK, smallest-algebraic mode) \cite{lehoucq1998arpack}, as
described in Sec.~\ref{sec:numerics}; for larger grids a shift-and-invert sparse
eigensolver (SLEPc/PETSc) \cite{hernandez2005slepc,balay1997petsc} targets the same
low-lying band more economically.

\section{Spectral Descriptor Variance from Disorder Theory}
\label{app:variances}

For completeness, we collect the leading-order descriptor variances derived in the
main text. They split into two groups, according to which part of the disorder each
descriptor measures.

The energy-position descriptors follow from the decomposition
$E_{\mathrm{dom}} = \bar{V}_s + \delta E_{\mathrm{dom}}$
(Secs.~\ref{sec:Ecent_theory} and \ref{sec:correlations}). The centroid tracks the
optically filtered slow potential, whereas the dominant-peak energy carries the
additional trap-switching fluctuation $\sigma_\delta$; because both share the same
$\bar{V}_s$, it cancels in their difference:
\begin{align}
\Var(E_{\mathrm{cent}}) &\approx W_s^2,
\\
\Var(E_{\mathrm{dom}}) &\approx W_s^2 + \sigma_\delta^2,
\\
\Var(\Delta E_{cd}) &\approx \sigma_\delta^2,
\\
\Var(R_{\mathrm{HL}}) &\propto \sigma_\delta^2.
\end{align}
The fine-structure descriptors instead measure the trap loading within the optical
spot. They scale with the combination $\Gamma_t = n_t W_t^2 \sigma_t^2$
(Sec.~\ref{sec:phase_diagram}), reduced by the spot area $\sigma_{\mathrm{opt}}^2$
(Sec.~\ref{sec:descriptor_theory}):
\begin{align}
\Var(F_S) &\propto n_t W_t^2 \sigma_t^2 / \sigma_{\mathrm{opt}}^2,
\label{eq:VarFS_etc}
\\
\Var(R_1) &\propto n_t W_t^2 \sigma_t^2 / (\eta^2 \sigma_{\mathrm{opt}}^2),
\\
\Var(S_{\mathrm{spec}}) &\propto n_t W_t^2 \sigma_t^2 / \sigma_{\mathrm{opt}}^2.
\end{align}
These relations give the scaling of the descriptor variances with the disorder
parameters, providing additional testable predictions beyond the spatial
correlation lengths.
For the dimensionless descriptors ($R_{\mathrm{HL}}$, $F_S$, $S_{\mathrm{spec}}$)
and for $R_1$ (of dimension inverse energy), the proportionality constants carry
the residual dimensions---fixed by the spectral envelope width and the broadening
$\eta$---so these last four relations express the parameter scaling rather than
dimensionally complete identities.

\section{Robustness to the Trap Width $\sigma_t$}
\label{app:sigma_robust}

To verify that the effective trap width $\sigma_t$ does not control the main
results---as argued in Sec.~\ref{sec:model}, where $\sigma_t$ enters only through
the degenerate combination $n_t W_t^2 \sigma_t^2$ and through the
non-binding hierarchy floor $\xi_t = \sqrt{2}\,\sigma_t$---we repeat the full
Hamiltonian diagonalization (Sec.~\ref{sec:numerics}) while sweeping $\sigma_t$
over $50$--$125~$nm, holding the disorder realization (slow field, trap positions,
and trap depths) fixed through common random seeds.
Table~\ref{tab:sigma_robust} reports the outcome in two modes.
In the \emph{uncompensated} sweep ($n_t$ fixed), increasing $\sigma_t$ deepens
and broadens the wells, raising the trap loading (and hence $\langle F_S\rangle$
and the trap-induced variance), yet the correlation-length hierarchy
$\xi(E_{\mathrm{cent}}) > \xi(E_{\mathrm{dom}})$ and the two principal
shape correlations $\rho_{\mathrm{S}}(\Delta E_{cd}, R_{\mathrm{HL}}) \approx -0.95$ and
$\rho_{\mathrm{S}}(F_S, R_1) \approx +0.93$ are unchanged.
The $E_{\mathrm{cent}}$--$E_{\mathrm{dom}}$ correlation, by contrast, is itself
loading-dependent and is not stabilized in this uncompensated sweep; its
anomalously small value at $\sigma_t = 125~$nm ($\rho_{EE} = -0.01$) reflects
overloading of the broadened trap wells and lies outside the best-fit regime,
and is included only as a stress test.
In the \emph{compensated} sweep, where $n_t \propto \sigma_t^{-2}$ holds the
combination $n_t \sigma_t^2$ fixed, the loading-dependent quantities
$\langle F_S\rangle$ and $\rho_{\mathrm{S}}(E_{\mathrm{cent}}, E_{\mathrm{dom}})$ also stabilize,
directly confirming the $n_t$--$\sigma_t$ degeneracy.
In all eight cases the hierarchy ordering is preserved; only the magnitude of the
gap narrows as $\sigma_t$ grows (the trap-induced fluctuation becomes spatially
broader, lengthening $\xi(E_{\mathrm{dom}})$ toward $\xi(E_{\mathrm{cent}})$),
never inverting.

\begin{table}[t]
\centering
\begin{ruledtabular}
\begin{tabular}{rcccccc}
$\sigma_t$ & $n_t$ & $\xi_{\mathrm{cent}}$ & $\xi_{\mathrm{dom}}$
 & $\rho_{cd}$ & $\rho_{FR}$ & $\rho_{EE}$ \\
(nm) & ($\mu$m$^{-2}$) & ($\mu$m) & ($\mu$m) & & & \\
\colrule
\multicolumn{7}{l}{\emph{Uncompensated} ($n_t = 2.0$ fixed)} \\
$50$  & $2.00$ & $1.18$ & $0.69$ & $-0.97$ & $+0.94$ & $+0.51$ \\
$75$  & $2.00$ & $1.36$ & $1.04$ & $-0.93$ & $+0.96$ & $+0.70$ \\
$100$ & $2.00$ & $1.16$ & $0.87$ & $-0.96$ & $+0.94$ & $+0.45$ \\
$125$ & $2.00$ & $1.21$ & $1.07$ & $-0.98$ & $+0.93$ & $-0.01$ \\
\colrule
\multicolumn{7}{l}{\emph{Compensated} ($n_t \propto \sigma_t^{-2}$)} \\
$50$  & $2.00$ & $1.18$ & $0.69$ & $-0.97$ & $+0.94$ & $+0.51$ \\
$75$  & $0.89$ & $0.98$ & $0.71$ & $-0.96$ & $+0.92$ & $+0.62$ \\
$100$ & $0.50$ & $1.15$ & $1.03$ & $-0.98$ & $+0.85$ & $+0.82$ \\
$125$ & $0.32$ & $1.18$ & $1.11$ & $-0.98$ & $+0.86$ & $+0.65$ \\
\end{tabular}
\end{ruledtabular}
\caption{Robustness of the diagonalization results to the trap width $\sigma_t$.
$\rho_{cd} \equiv \rho_{\mathrm{S}}(\Delta E_{cd}, R_{\mathrm{HL}})$,
$\rho_{FR} \equiv \rho_{\mathrm{S}}(F_S, R_1)$, and
$\rho_{EE} \equiv \rho_{\mathrm{S}}(E_{\mathrm{cent}}, E_{\mathrm{dom}})$.
The hierarchy $\xi(E_{\mathrm{cent}}) > \xi(E_{\mathrm{dom}})$ holds in every row.
The compensated sweep keeps $n_t \sigma_t^2$ fixed, stabilizing the
loading-dependent observables $\langle F_S\rangle$ and $\rho_{EE}$, and thereby
confirming the $n_t W_t^2 \sigma_t^2$ degeneracy.}
\label{tab:sigma_robust}
\end{table}

\section{Robustness to the Optical-Weight Model}
\label{app:optical_weight}

The local PL weight of a COM eigenstate is, physically, the coherent radiative
(oscillator-strength) matrix element $A_n(\mathbf{R}) = |\!\int\! W_{\mathrm{amp}}(\mathbf{r}-\mathbf{R})\psi_n(\mathbf{r})\,d^2\mathbf{r}|^2$
of Eq.~\eqref{eq:A_weight}
\cite{runge1998spatially,runge1998level,zimmermann2003theory}, rather than the
incoherent local density $\int W_{\mathrm{opt}}|\psi_n|^2$ used in the transparent
LDOS form \eqref{eq:I_formal}. Here we verify that the descriptor statistics are
insensitive to this choice. Using the same Hamiltonian eigenstates
(Sec.~\ref{sec:key_results}), we recompute all descriptors with both weightings
for the same $128\times128$, $8\times8~\mu$m$^2$ realizations, averaged over five
independent disorder seeds. As shown in Fig.~\ref{fig:optical_weight}, every
descriptor statistic---$\xi(E_{\mathrm{cent}})$, $\xi(E_{\mathrm{dom}})$, their
ratio, the three principal Spearman correlations, and $\langle F_S\rangle$---agrees
between the two weightings to well within the seed-to-seed scatter
(the coherent and incoherent values differ by $\le 0.01$ in every case, e.g.\
$\rho_{\mathrm{S}}(\Delta E_{cd},R_{\mathrm{HL}}) = -0.928$ vs.\ $-0.929$ and
$\xi(E_{\mathrm{cent}}) = 1.143$ vs.\ $1.142~\mu$m). The reason is structural:
both weights are localized within the optical spot, so the spatial statistics of
the descriptor fields---and hence the correlation-length hierarchy and descriptor
covariance that are the subject of this work---are set by the disorder landscape,
not by the detailed form of the per-state emission weight. The coherent weight
does redistribute intensity toward the nodeless localized (trap) states, but this
common re-weighting cancels in the descriptor \emph{correlations}.

\begin{figure}[t]
\centering
\includegraphics[width=0.72\columnwidth]{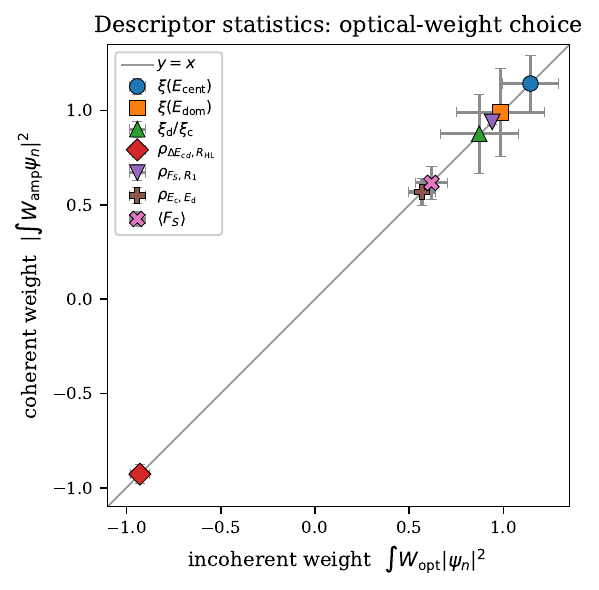}
\caption{
\textbf{Descriptor statistics are insensitive to the optical-weight model.}
For each descriptor statistic (labeled), the value obtained with the coherent
radiative weight $|\!\int\! W_{\mathrm{amp}}\psi_n|^2$ (vertical axis) is plotted
against that from the incoherent local-density weight
$\int W_{\mathrm{opt}}|\psi_n|^2$ (horizontal axis); error bars are the
seed-to-seed standard deviation over five disorder realizations. All points lie
close to the $y=x$ line, showing that the two weightings yield statistically
indistinguishable descriptor covariance and correlation-length hierarchy within
the seed-to-seed uncertainty.
}
\label{fig:optical_weight}
\end{figure}

\section{Finite-Size, Finite-Step, and Finite-Spot Robustness of the Hierarchy}
\label{app:finite_size}

The experimental map spans $8\times8~\mu$m$^2$, only a few times the correlation
lengths, so the extracted $1/e$ lengths carry finite-size uncertainty. We test the
robustness of the hierarchy $\xi(E_{\mathrm{cent}}) > \xi(E_{\mathrm{dom}})$ by
sweeping the map size $L$, the scan pitch, and the optical-spot width, each over
thirty disorder realizations, using the fast phenomenological descriptor generator
(Sec.~\ref{sec:key_results}) so that many configurations can be sampled.
Figure~\ref{fig:finite_size} summarizes the domain-size sweep, showing the
correlation lengths $\xi(E_{\mathrm{cent}})$ and $\xi(E_{\mathrm{dom}})$ versus
map size [Fig.~\ref{fig:finite_size}(a)] and the finite-size-robust
hierarchy ratio [Fig.~\ref{fig:finite_size}(b)]. Two features stand out.
First, in this phenomenological sweep the hierarchy holds across the sampled
configurations, and the ratio $\xi(E_{\mathrm{dom}})/\xi(E_{\mathrm{cent}})$ is much
less sensitive to sampling than the absolute lengths:
it stays at $0.69$--$0.70$ across
$L = 8$--$20~\mu$m and scan pitches $0.25$--$0.60~\mu$m, and varies only mildly
(from $0.74$ to $0.64$) as the optical FWHM is swept over $0.60$--$1.10~\mu$m,
all bracketing the experimental ratio $0.64$. Second, the \emph{absolute} $1/e$
lengths are biased low on the smallest map and drift upward with $L$ (a
finite-size estimator bias); at the experimental map size the simulation
reproduces the measured values ($\xi(E_{\mathrm{cent}}) = 1.92\pm0.17~\mu$m vs.\ $2.00$;
$\xi(E_{\mathrm{dom}}) = 1.33\pm0.25~\mu$m vs.\ $1.27$). These differ slightly from
the Fig.~\ref{fig:correlation_hierarchy} ensemble ($1.93$ and $1.38~\mu$m) because
the finite-size sweep uses an independently generated set of realizations. Because
finite sampling biases both lengths in the same direction, the physically
meaningful ratio is comparatively stable. The principal anticorrelation is likewise robust,
$\rho_{\mathrm{S}}(\Delta E_{cd}, R_{\mathrm{HL}}) = -0.97\pm0.01$ throughout.
This supports the conclusion that the correlation-length hierarchy is not primarily
a finite-sampling or coarse-gridding artifact, consistent with the ensemble-average
inequality of Proposition~\ref{prop:hierarchy}; occasional per-realization inversions seen in
the microscopic diagonalization on a single $8~\mu$m map are the expected
finite-size scatter about this ensemble result.

\begin{figure}[t]
\centering
\includegraphics[width=\columnwidth]{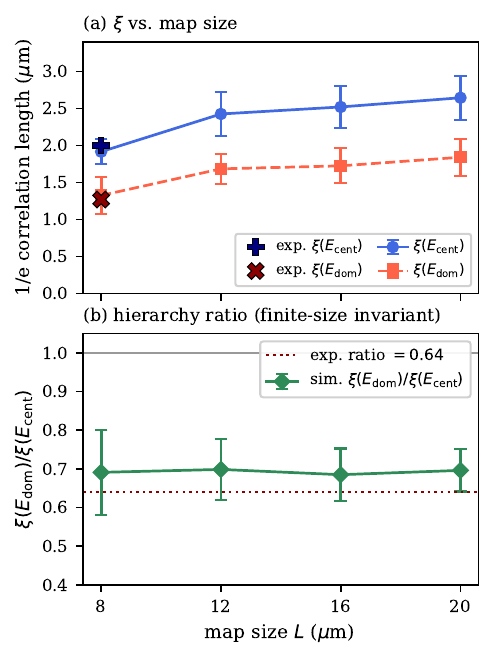}
\caption{
\textbf{Finite-size behavior of the correlation-length hierarchy.}
(a) $1/e$ correlation lengths $\xi(E_{\mathrm{cent}})$ and $\xi(E_{\mathrm{dom}})$
versus map size $L$ (mean $\pm$ std over thirty disorder realizations; fixed
$0.4~\mu$m pitch and $0.85~\mu$m optical FWHM); filled markers are the
experimental values at $L = 8~\mu$m.
(b) The hierarchy ratio $\xi(E_{\mathrm{dom}})/\xi(E_{\mathrm{cent}})$ is
finite-size-robust ($\approx 0.69$--$0.70$) and close to the experimental
value $0.64$ (dotted line), even though the absolute lengths in (a) drift with
$L$.
}
\label{fig:finite_size}
\end{figure}

\bibliographystyle{apsrev4-2}
\bibliography{refs}

\end{document}